\documentclass{amsart}
\usepackage{mathrsfs}
\usepackage{graphicx}
\usepackage{amsmath}
\usepackage{amscd}
\usepackage{cite}
\usepackage{hyperref}
\usepackage{epsfig}
\usepackage{subfigure}
\usepackage{caption}
\usepackage{tikz}
  \usepackage{float}
  \usepackage{amssymb}
\usepackage{amsfonts}
\usepackage[all,cmtip]{xy}
\usepackage{appendix}
\usepackage{bm}

\newtheorem{Example}{Example}[section]

\newtheorem{Theorem}{Theorem}[section]
\newtheorem{Theorem/Definition}{Theorem/Definition}[section]
\newtheorem{Proposition}{Proposition}[section]
\newtheorem{Lemma}{Lemma}[section]
\newtheorem{Corollary}{Corollary}[section]

\newcommand{\pd}{\partial}

\newcommand{\bC}{{\mathbb C}}

\newcommand{\bZ}{{\mathbb Z}}

\newcommand{\cB}{{\mathcal B}}
\newcommand{\cC}{{\mathcal C}}

\newcommand{\cF}{{\mathcal F}}

\newcommand{\cL}{{\mathcal L}}
\newcommand{\cM}{{\mathcal M}}

\newcommand{\half}{\frac{1}{2}}

\newcommand{\Mbar}{\overline{\cM}}

\newcommand{\be}{\begin{equation}}
\newcommand{\ee}{\end{equation}}
\newcommand{\bea}{\begin{eqnarray}}
\newcommand{\ben}{\begin{eqnarray*}}
\newcommand{\een}{\end{eqnarray*}}
\newcommand{\eea}{\end{eqnarray}}

\DeclareMathOperator{\Id}{id}

\DeclareMathOperator{\Res}{Res}

\usepackage{color}

\definecolor{yellow}{rgb}{1,1,0}
\definecolor{orange}{rgb}{1,.7,0}
\definecolor{red}{rgb}{1,0,0}
\definecolor{green}{rgb}{0,1,1}
\definecolor{white}{rgb}{1,1,1}

\definecolor{A}{rgb}{.75,1,.75}

\theoremstyle{remark}
\newtheorem{Remark}{Remark}[section]

\begin{document}

\newtheorem{myDef}{Definition}
\newtheorem{thm}{Theorem}
\newtheorem{eqn}{equation}

\title[Affine Coordinates for Open Intersection Numbers]
{On affine coordinates of the tau-function for open intersection numbers}

\author{Zhiyuan Wang}
\address{School of Mathematical Sciences\\
Peking University\\Beijing, 100871, China}
\email{zhiyuan19@math.pku.edu.cn}

\begin{abstract}
In Alexandrov's work \cite{al2, al3} it has been shown that
the extended partition function $\exp(F^{o,ext}+F^c)$ introduced by Buryak in \cite{bu, bu2}
is a tau-function of the KP hierarchy.
In this work,
we compute the affine coordinates of this tau-function on the Sato Grassmannian,
and rewrite the Virasoro constraints as recursions for the affine coordinates
in the fermionic picture.
As applications
we derive some formulas for the extended partition function and
the connected $n$-point functions using methods developed by Zhou in \cite{zhou1}
based on the boson-fermion correspondence.

\end{abstract}

\maketitle

%%\tableofcontents

\section{Introduction}

Integrable hierarchies have become one of the central topics in the study of mirror symmetry
since the famous Witten Conjecture/Kontsevich Theorem \cite{wi, ko},
which asserts that certain generating series of the intersection numbers of $\psi$-classes
on the moduli spaces $\Mbar_{g,n}$ of stable curves is a tau-function $\tau^{\text{WK}}$ of the
Korteweg-de Vries (KdV) hierarchy.
It is well-known that the Kadomtsev-Petviashvili (KP) hierarchy is a more general integrable system,
and the KdV hierarchy is a reduction of the KP hierarchy.
See \cite{kp} and \cite{KdV} for the KP equation and KdV equation respectively,
and \cite{djm, jm} for an introduction to Kyoto School's approach to the KP hierarchy.

In Kyoto School's approach,
there are three equivalent descriptions of a tau-function of the KP hierarchy
(see \cite{djm, sa}):
\begin{itemize}
\item[1)]
A tau-function corresponds to an element of the Sato Grassmannian;
\item[2)]
A tau-function is a vector in the fermionic Fock space $\cF$
(satisfying the Hirota bilinear relation);
\item[3)]
A tau-function is a vector in the bosonic Fock space $\cB$
(satisfying the Hirota bilinear relation).
\end{itemize}
In particular,
given an element $U\in Gr_{(0)}$ in the big cell of the Sato Grassmannian,
one can associate a vector in each Fock space:
\be
\label{eq-intro-1}
U\in Gr_{(0)}
\quad\to\quad
|U\rangle \in \cF^{(0)}
\quad\to\quad
\tau_U\in\Lambda,
\ee
where $\cF^{(0)}\subset \cF$ is the subspace consisting of states of charge $0$,
and $\Lambda \subset \cB$ is the space of symmetric functions.
The second step $|U\rangle \mapsto \tau_U$ is the so-called boson-fermion correspondence.

One way to describe the first step $U\mapsto |U\rangle$ is to use Bogoliubov transformations
of the vacuum on the fermionic Fock space.
This was known by experts since the works of the Kyoto School,
and has been used many times in physics literatures for various models,
see \cite{adkmv, akmv} and \cite{kmmmz} for examples,
and see also \cite{az} for a review.
In particular,
the ADKMV conjecture \cite{adkmv, akmv} gives a simple expression for the topological vertex
as a Bogoliubov transform on the $3$-component fermionic Fock space.
See \cite{dz1, dz2} for a generalization of the ADKMV conjecture and the proofs for the one-legged and two-legged cases.
Inspired by these results,
Zhou found an explicit formula for the Witten-Kontsevich tau-function $\tau^{\text{WK}}$
as a Bogoliubov transform of the fermionic vacuum in \cite{zhou4}.
His method is to rewrite the Virasoro constraints for $\tau^{\text{WK}}$ (see \cite{dvv, fkn}) in terms of the free fermions,
and then derive recursion relations to compute the coefficients of the Bogoliubov transformation.
The explicit expressions of the coefficients are rather complicated,
but can be interpreted in a nice way,
due to the work \cite{by} of Balogh and Yang.
They found that these coefficients in Zhou's formula are the affine coordinates for the tau-function $\tau^{\text{WK}}$
on the Sato Grassmannian.
Later in \cite{zhou1},
Zhou generalized this idea to the case of an arbitrary tau-function of the KP hierarchy,
and proposed a framework called emergent geometry.
Moreover,
he proved a new formula to compute the connected $n$-point functions
associated to a KP tau-function using the affine coordinates in that work.
See \cite{zhou7, zhou3} for the computation of the affine coordinates of
some other well-known tau-functions in mathematical physics,
including the enumeration of Grothendieck's dessin d'enfants,
the Hermitian one-matrix models,
and the generalized BGW model.
In these examples,
one can write down an explicit formula for the tau-function as a Bogoliubov transform
in the fermionic picture
and derive formulas for the (all genera) connected $n$-point functions,
while finding explicit formulas in the bosonic picture seems difficult.

For each tau-function $\tau$ of the KP hierarchy,
there is a unique way to represent it as the following Bogoliubov transform of the fermionic vacuum
(see \cite[\S 3]{zhou1}):
\ben
\tau = \exp\bigg(\sum_{n,m\geq 0} a_{n,m} \psi_{-m-\half} \psi_{-n-\half}^*\bigg) |0\rangle,
\een
where $\{a_{n,m}\}_{n,m\geq 0}$ are the affine coordinates for $\tau$ on the Sato Grassmannian,
and $\psi_{-m-\half}, \psi_{-n-\half}^*$ are the fermionic creators (for $n,m\geq 0$).
This representation of $\tau$ allows us to
translate some constraints for $\tau$ into the fermionic picture as constraints for $\{a_{n,m}\}_{n,m\geq 0}$.
It is well-known that on the Sato Grassmannian,
a good way to select a particular tau-function is to use Kac-Schwarz operators (see \cite{ks, sc});
and in the bosonic picture,
one usually uses differential equations
(such as the string equation, Virasoro constraints, etc.)
to select a tau-function.
Now in the fermionic Fock space,
a way to select a tau-function will be imposing certain recursion relations or constraints
for the affine coordinates $\{a_{n,m}\}_{n,m\geq 0}$.

In this work,
we study the open intersection numbers
developed by Pandharipande, Solomon, and Tessler in \cite{pst}.
This is the open analog of the intersections on $\Mbar_{g,n}$ that appeared in the Witten conjecture.
In the work \cite{pst},
these authors investigated the intersection theory
on the moduli spaces of Riemann surfaces with boundary.
They constructed rigorously the moduli spaces and open intersection numbers in genus zero,
and conjectured that the all-genera generating function
is uniquely determined by the open Virasoro constraints and
a family of differential equations called the open KdV hierarchy.
In \cite{bu2},
Buryak proved that the open KdV is equivalent to the open Virasoro constraints
for this conjectural all-genera intersection theory,
and there is a unique solution $F^o$ (called the open potential) to these constraints
specified by a certain initial condition.
Moreover,
he constructed an extended open potential $F^{o,ext}$ to include the descendants on the boundary,
and derived the Virasoro constraints for the extended generating function
\be
\tau^o := \exp(F^{o,ext} + F^c)
\ee
in \cite{bu},
where $F^c=\log \tau^{\text{WK}}$ is the free energy of the Witten-Kontsevich tau-function.
He also discussed some relations between $F^{o}$ and $F^c$.
In \cite{al2,al3},
Alexandrov found that the extended generating function $\tau^o$ is the special case $N=1$
of a family of Kontsevich-Penner matrix integral $\tau_N$,
while another special case $N=0$ gives the Witten-Kontsevich tau-function.
Moreover,
he showed that $\tau_N$ is a tau-function of the KP hierarchy,
and constructed an admissible basis
(denoted by $\Phi_i^N$ in \cite{al2}, expressed in terms of some integrals)
on the Sato Grassmannian together with
a Kac-Schwarz operator for each $N$.
The Virasoro constraints and W-constraints for $\tau_N$
were also constructed in his works.
See Buryak-Tessler \cite{bt} for a proof of the conjecture of \cite{pst}
(where the construction of the higher genus intersection theory was announced by Solomon and Tessler,
see also \cite{te}).
For more about the open intersection theory,
see \cite{br, abt, bey, ke, bh, saf, al4, liu, saf2, amr}.

Now we briefly summarize the main results of the present work.
Our method to compute the affine coordinates of the tau-function $\tau^o$ is to
reformulate the Virasoro constraints for $\tau^o$ in terms of the fermionic operators
and then derive recursions for the affine coordinates,
following the strategy developed in \cite{zhou3, zhou4}.
The result is (see Theorem \ref{thm-mainthm}):
\begin{Theorem}
The affine coordinates $\{a_{n,m}^o\}_{n,m\geq 0}$ for the tau-function $\tau^o$ are:
\be
a_{0,3q-1}^o=\big(\frac{3}{2}\big)^q \cdot (2q-1)!!\cdot \sum_{j=0}^q
\frac{(6j-1)!!}{54^j\cdot (2j)!\cdot (2j-1)!!},
\ee
and for every $p\geq 0$ and $q\geq 1$,
\be
\begin{split}
& a_{3p+1,3q-2}^o = a_{3p,3q-1}^{\text{WK}},\\
& a_{3p+2,3q-3}^o = a_{3p+1,3q-2}^{\text{WK}},\\
& a_{3p+3,3q-1}^o = a_{3p+2,3q}^{\text{WK}} - a_{0,3q-1}^o \cdot a_{3p+2,0}^{\text{WK}},\\
\end{split}
\ee
where $\{a_{n,m}^{\text{WK}}\}_{n,m\geq 0}$ are the affine coordinates of the Witten-Kontsevich tau-function
whose explicit formulas have been computed by Zhou in \cite{zhou4}
(see \eqref{eq-WK-coeff} for the formulas).
Moreover,
we have $a_{n,m}^o=0$ for $n+m\not\equiv -1 (\text{mod }3)$.
\end{Theorem}

We will show that the fermionic reformulation of the Virasoro constraints for $\tau^o$
is equivalent to the following recursions and constraints for the affine coordinates
(see Theorem \ref{thm-coeff-rec-1} and Theorem \ref{thm-coeff-rec-2}):
\begin{Theorem}
For every $n\geq -1$,
the affine coordinates $\{a_{n,m}^o\}_{n,m\geq 0}$ for the tau-function $\tau^o$ satisfy:
\begin{itemize}
\item[1)]
A $(2n+3)$-step recursion relation:
\begin{equation*}
\begin{split}
a_{l,2n+3+m}^o - &a_{2n+3+l,m}^o
=(2n+\frac{5}{2}+m) a_{l,2n+m}^o
+(l-\frac{3}{2}) a_{2n+l,m}^o\\
&\qquad
+\sum_{k=0}^{2n+2} a_{k,m}^o a_{l,2n+2-k}^o
-\sum_{k=0}^{2n-1} (2n+\frac{3}{2}-k) a_{k,m}^o a_{l,2n-1-k}^o,\\
\end{split}
\end{equation*}
\item[2)]
A linear constraint:
\begin{equation*}
\sum_{k=0}^{2n+2} a_{k,2n+2-k}^o =
\sum_{k=0}^{2n-1} (2n+\frac{3}{2}-k) a_{k,2n-1-k}^o.
\end{equation*}
\end{itemize}
Here we set $a_{n,m}^o:=0$ for $n<0$ or $m<0$.
\end{Theorem}

As a corollary,
we obtain a fermionic representation of $\tau^o$ as a Bogoliubov transform
of the fermionic vacuum:
\be
\tau^o = \exp (A^o) |0\rangle,
\ee
where $A^o$ is a quadratic operator of the fermionic creators:
\ben
A^o = \sum_{n,m\geq 0} a_{n,m}^o \psi_{-m-\half}\psi_{-n-\half}^*.
\een
This leads to a formula in the bosonic picture
in terms of Schur functions
(see \S \ref{sec-f-bogoliubov}-\S \ref{sec-b-schur} for details):
\be
\tau^o = 1+
\sum_{\mu:\text{ } |\mu|>0}
(-1)^{n_1+\cdots+n_k}\cdot\det(a_{n_i,m_j}^o)_{1\leq i,j\leq k} \cdot s_\mu,
\ee
where $\mu=(m_1,\cdots,m_k|n_1,\cdots,n_k)$ is the Frobenius notation of a partition $\mu$.
The Schur functions $s_\mu$ are related to the Newton symmetric functions $p_k$ by (see \cite{mac}):
\ben
p_\mu = \sum_{\lambda} \chi_\mu^\lambda s_\lambda,
\een
where $p_\mu=p_{\mu_1}\cdots p_{\mu_k}$ for $\mu=(\mu_1,\cdots,\mu_k)$,
and in the boson-fermion correspondence the
time variables of the KP hierarchy are $T_n := \frac{p_n}{n}$
(see \S \ref{sec-symm-fn} and \S \ref{sec-b-f} for a brief review).
The above summation is over all partitions $\lambda$ with $|\lambda|=|\mu|$,
and $\chi_\mu^\lambda$ are the characters of irreducible representations of the symmetric group $S_{|\mu|}$.

Furthermore,
we compute the generating series
\ben
A^o (x,y):= \sum_{n,m\geq 0} a_{n,m}^o x^{-n-1} y^{-m-1}
\een
of the affine coordinates for $\tau^o$, and prove (see Theorem \ref{thm-generating}):
\begin{Theorem}
We have:
\be
A^o(x,y)= \frac{1}{y-x} +
\frac{c(y)a(-x)}{x}
+\frac{y}{x}\cdot
\frac{a(y)b(-x)-a(-x)b(y)}{y^2-x^2},
\ee
where $a(z),b(z)$ are the Faber-Zagier series:
\begin{equation*}
a(z)=\sum_{m=0}^\infty \frac{(6m-1)!!}{36^m\cdot (2m)!}z^{-3m},
\qquad
b(z)=-\sum_{m=0}^\infty \frac{(6m-1)!!}{36^m\cdot (2m)!}\frac{6m+1}{6m-1}z^{-3m+1},
\end{equation*}
and $c(z)$ is the generating series of $\{a_{0,m}^o\}_{m\geq 0}$:
\ben
c(z)=\sum_{m\geq 0}
\big(\frac{3}{2}\big)^m \cdot (2m-1)!!\cdot \sum_{j=0}^m
\frac{(6j-1)!!}{54^j\cdot (2j)!\cdot (2j-1)!!}\cdot
z^{-3m}.
\een
\end{Theorem}
Having this formula for the generating function,
we are able to compute the connected $n$-point functions for $\log(\tau^o)$ and $F^{o,ext}$
using a formula derived by Zhou in \cite{zhou1}.
See \S \ref{sec-gen}-\S \ref{sec-relation-npt-3} for details of the computations.
In particular,
the one-point function of $F^{o,ext}$ is given by
(see \S \ref{sec-relation-npt-1} for details):
\begin{equation*}
\sum_{j\geq 1}
\frac{\pd F^{o,ext}}{\pd T_j} \bigg|_{\bm T=0}
\cdot z^{-j-1}
=\frac{c(z)a(-z)-1}{z}.
\end{equation*}
And the two-point function is given by (see \S \ref{sec-relation-npt-2}):
\begin{equation*}
\begin{split}
&\sum_{j,k\geq 1}\frac{\pd^2 F^{o,ext}}{\pd T_j\pd T_k}\bigg|_{\bm T =0} \cdot
x^{-j-1} y^{-k-1}=
- \frac{a(-x)a(-y)c(x)c(y)}{xy}\\
&\qquad\qquad\qquad\qquad +\frac{1}{x(x^2-y^2)}\bigg(
a(y)a(-y)b(-x)c(x) - a(-x)a(-y)b(y)c(x)
\bigg)\\
&\qquad\qquad\qquad\qquad +\frac{1}{y(x^2-y^2)}\bigg(
a(-x)a(-y)b(x)c(y) - a(x)a(-x)b(-y)c(y)
\bigg).
\end{split}
\end{equation*}
See also \S \ref{sec-relation-npt-3} for a similar but more complicated formula for the $3$-point function.
It will be interesting to compare these formulas with the formulas for the $n$-point functions
derived by Bertola-Ruzza in \cite{br} using a different method.

A further question is to compute the affine coordinates of the partition function
$\tau_N$ for the Kontsevich-Penner model (see \cite{al2}) for general $N$,
and we hope to address this problem in future works.
It is also interesting to discuss the emergent geometry for $\tau^o$ and $\tau_N$ following
the methods developed in \cite{zhou1, zhou8, zhou9},
and this will fit into the general picture listed in \cite[p.5]{zhou3} which describes the dualities
among various theories whose partition functions are all tau-functions of the KP hierarchy.

The rest of this paper is arranged as follows.
In \S \ref{sec-pre-KP} we recall some preliminaries of Kyoto School's approach to the KP hierarchy.
In \S \ref{sec-affinecoord} we give the main results on the affine coordinates of the tau-function $\tau^o$.
The proofs of these results will be given in \ref{sec-proof},
where we derive a fermionic description of the Virasoro constraints for $\tau^o$.
In \S \ref{sec-applications} we give some applications of the main theorem,
including the computations of $\tau^o$ and $\log(\tau^o)$ using the affine coordinates.

\section{Preliminaries of KP Hierarchy and Sato's Theory}
\label{sec-pre-KP}

We recall some basic knowledge of Kyoto School's approach to the KP hierarchy in this section.
In particular,
we describe the constructions in \eqref{eq-intro-1}.

\subsection{Partitions of integers, Young diagrams, and Schur functions}
\label{sec-symm-fn}

First we recall the partitions of integers and Schur functions.
See \cite{mac} for an introduction.

A partition of a positive integer $n$ is a sequence of integers
$\lambda=(\lambda_1,\lambda_2,\cdots,\lambda_l)$ satisfying
$\lambda_1\geq \cdots\geq \lambda_l> 0$ and $|\lambda|:=\lambda_1+\cdots+\lambda_l = n$.
The number $l$ is called the length of the partition $\lambda$,
denoted by $l(\lambda)$.

There is a one-to-one correspondence between the set of all partitions and the set of Young diagrams.
Given a partition $\lambda=(\lambda_1,\cdots,\lambda_l)$,
we associate a Young diagram consisting of $|\lambda|$ boxes,
such that there are $\lambda_i$ boxes on the $i$-th row.
For example,
all the partitions of $4$ and their associated Young diagrams are:
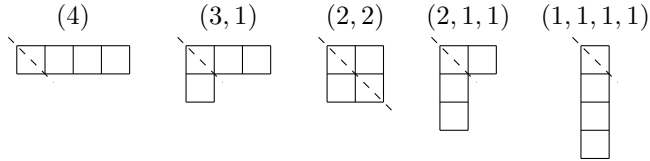
\begin{equation*}
\begin{tikzpicture}[scale=0.75]
\draw (0,0) -- (2,0);
\draw (0,-0.5) -- (2,-0.5);
\draw (0,0) -- (0,-0.5);
\draw (0.5,0) -- (0.5,-0.5);
\draw (1,0) -- (1,-0.5);
\draw (1.5,0) -- (1.5,-0.5);
\draw (2,0) -- (2,-0.5);
\draw [dashed] (-0.15,0.15) -- (0.65,-0.65);
\node [align=center,align=center] at (1,0.5) {$(4)$};
\draw (3,0) -- (4.5,0);
\draw (3,-0.5) -- (4.5,-0.5);
\draw (3,-1) -- (3.5,-1);
\draw (3,0) -- (3,-1);
\draw (3.5,0) -- (3.5,-1);
\draw (4,0) -- (4,-0.5);
\draw (4.5,0) -- (4.5,-0.5);
\draw [dashed] (3-0.15,0.15) -- (3.65,-0.65);
\node [align=center,align=center] at (3.75,0.5) {$(3,1)$};
\draw (5.5,0) -- (6.5,0);
\draw (5.5,-0.5) -- (6.5,-0.5);
\draw (5.5,-1) -- (6.5,-1);
\draw (5.5,0) -- (5.5,-1);
\draw (6,0) -- (6,-1);
\draw (6.5,0) -- (6.5,-1);
\draw [dashed] (5.5-0.15,0.15) -- (6.65,-1.15);
\node [align=center,align=center] at (6,0.5) {$(2,2)$};
\draw (5.5+2,0) -- (6.5+2,0);
\draw (5.5+2,-0.5) -- (6.5+2,-0.5);
\draw (5.5+2,-1) -- (6+2,-1);
\draw (5.5+2,-1.5) -- (6+2,-1.5);
\draw (5.5+2,0) -- (5.5+2,-1.5);
\draw (6+2,0) -- (6+2,-1.5);
\draw (6.5+2,0) -- (6.5+2,-0.5);
\draw [dashed] (7.5-0.15,0.15) -- (7.5+0.65,-0.65);
\node [align=center,align=center] at (6+2,0.5) {$(2,1,1)$};
\draw (8.5+2,0) -- (8+2,0);
\draw (8.5+2,-0.5) -- (8+2,-0.5);
\draw (8.5+2,-1) -- (8+2,-1);
\draw (8.5+2,-1.5) -- (8+2,-1.5);
\draw (8.5+2,-2) -- (8+2,-2);
\draw (8+2,0) -- (8+2,-2);
\draw (8.5+2,0) -- (8.5+2,-2);
\draw [dashed] (10-0.15,0.15) -- (10+0.65,-0.65);
\node [align=center,align=center] at (8.25+2,0.5) {$(1,1,1,1)$};
\end{tikzpicture}
\end{equation*}

The dotted lines above are called the diagonal of the Young diagram.
Given a partition $\lambda=(\lambda_1,\cdots,\lambda_l)$,
denote by $\lambda^t=(\lambda_1^t,\cdots,\lambda_m^t)$ its transpose
(obtained by flipping the Young diagram along the diagonal), then clearly $m=\lambda_1$,
and $(\lambda^t)^t =\lambda$.
For example,
from the above expressions one easily sees that:
\begin{equation*}
(4)^t = (1,1,1,1),
\qquad
(3,1)^t = (2,1,1),
\qquad
(2,2)^t = (2,2).
\end{equation*}
The Frobenius notation of a partition $\lambda$ is defined to be:
\ben
\lambda= (m_1,m_2,\cdots,m_k | n_1,n_2,\cdots,n_k),
\een
where $k$ is the number of boxes on the diagonal of the associated Young diagram,
and $m_i := \lambda_i - i$,
$n_i:=\lambda_i^t - i$.

Now let us recall the definition of Schur functions $s_\lambda$
where $\lambda$ is a partition.
Fix a sequence of formal variables $\bm y:=(y_1,y_2,\cdots)$.
First for the special case $\lambda=(m|n)$,
define the Schur function $s_{(m|n)}$ by:
\be
s_{(m|n)}(\bm y)= h_{m+1}(\bm y)e_n(\bm y) - h_{m+2}(\bm y)e_{n-1}(\bm y) + \cdots
+ (-1)^n h_{m+n+1}(\bm y),
\ee
where $h_n$ and $e_n$ are the
complete symmetric function and elementary symmetric function of degree $n$ respectively:
\begin{equation*}
h_n(\bm y):= \sum_{\sum_{i} d_i =n} y_1^{d_1}y_2^{d_2}y_3^{d_3}\cdots,
\qquad\quad
e_n(\bm y):= \sum_{i_1<\cdots <i_n} y_{i_1}\cdots y_{i_n}.
\end{equation*}
Now let $\lambda=(m_1,\cdots,m_k|n_1,\cdots,n_k)$.
Then the Schur function $s_\lambda$ labeled by $\lambda=(m_1,\cdots,m_k|n_1,\cdots,n_k)$
is defined to be:
\be
s_\lambda (\bm y) = \det (s_{(m_i|n_j)} (\bm y))_{1\leq i,j\leq k}.
\ee
Let $\Lambda$ be the space of all symmetric functions in $\bm y$,
then all Schur functions $\{s_\lambda (\bm y)\}_\lambda$ form a basis for $\Lambda$
(here $s_{(0)}:=1$).

\begin{Remark}
Let $p_n$ be the Newton symmetric function defined by:
\ben
p_n (\bm y):=\sum_i y_i^n,
\een
and denote $p_\mu:=p_{\mu_1}\cdots p_{\mu_k}$.
Then $\{p_\mu\}_\mu$ and $\{s_\mu\}_\mu$ are two bases of $\Lambda$,
and they are related by
(see eg. \cite{mac}):
\be
\label{eq-newton-schur}
p_\mu = \sum_{\lambda} \chi_\mu^\lambda s_\lambda,
\qquad
\qquad
s_\mu = \sum_{\lambda} \frac{\chi_\lambda^\mu}{z_\lambda} p_\lambda,
\ee
where $z_\lambda= \prod_{i\geq 1} i^{k_i}\cdot k_i!$
and $k_i$ is the number of $i$ appearing in the partition $\lambda$.
The summations on the right-hand sides of \eqref{eq-newton-schur} are over all partitions $\lambda$
with $|\lambda|=|\mu|$,
and $\{\chi_\mu^\lambda\}$ are the characters of irreducible representations of
the symmetric group $S_{|\mu|}$ (see eg. \cite[\S 4]{fh} for an introduction of the representation theory of $S_n$).
\end{Remark}

\subsection{Semi-infinite wedge products and the fermionic Fock space}

Now we recall the construction of the fermionic Fock space $\cF$
via semi-infinite wedge product,
and the action of a Clifford algebra on $\cF$.

Let $\bm{a}=(a_1,a_2,\cdots)$ be a sequence of half-integers in $\bZ+\half$,
satisfying the condition $a_1<a_2<\cdots$.
We say $\bm{a}$ is admissible if:
\ben
\big|(\bZ_{\geq 0}+\half)-\{a_1,a_2\cdots\}\big|<\infty,
\qquad
\big|\{a_1,a_2\cdots\}-(\bZ_{\geq 0}+\half)\big|<\infty.
\een
Given such an $\bm{a}$,
one can associate a semi-infinite wedge product $|\bm a\rangle$ by:
\ben
| \bm a\rangle :=
z^{a_1} \wedge z^{a_2} \wedge z^{a_3} \wedge \cdots
\een
Then the fermionic Fock space $\cF$ is the space of expressions of the following form:
\ben
\sum_{\bm a} c_{\bm a} |\bm a\rangle,
\een
where the summation is taken over admissible sequences.

Assume $|\bm a\rangle\in \cF$.
The charge of the vector $|\bm a\rangle$ is defined to be the number:
\begin{equation*}
\text{charge}(|\bm a\rangle):=
\big|(\bZ_{\geq 0}+\half)-\{a_1,a_2\cdots\}\big|-
\big|\{a_1,a_2\cdots\}-(\bZ_{\geq 0}+\half)\big|,
\end{equation*}
and this gives us a decomposition of the fermionic Fock space:
\be
\cF=\bigoplus_{n\in \bZ} \cF^{(n)},
\ee
where $\cF^{(n)}$ is spanned by vectors $|\bm a\rangle$ of charge $n$.
Denote:
\ben
|n\rangle := z^{n+\half}\wedge z^{n+\frac{3}{2}}
\wedge z^{n+\frac{5}{2}} \wedge \cdots \in \cF^{(n)}.
\een
In particular, the vector
$|0\rangle := z^{\half}\wedge z^{\frac{3}{2}}
\wedge z^{\frac{5}{2}} \wedge \cdots \in \cF^{(0)}$
is called the fermionic vacuum vector.

The space $\cF^{(0)}$ has a basis labeled by partitions of integers.
Let $\mu=(\mu_1,\mu_2,\cdots)$ be a partition such that
$\mu_1\geq \mu_2\geq \cdots\geq \mu_l >\mu_{l+1}=\mu_{l+2}=\cdots=0$,
and denote:
\ben
|\mu\rangle:=
z^{\frac{1}{2}-\mu_1}\wedge z^{\frac{3}{2}-\mu_2}\wedge
z^{\frac{5}{2}-\mu_3}\wedge \cdots \in\cF^{(0)},
\een
then $\{|\mu\rangle\}$ form a basis for $\cF^{(0)}$.
In particular, $(0,0,\cdots)$ is a partition of $0\in \bZ$,
and the corresponding vector is $|0\rangle\in\cF^{(0)}$.

Now we recall the action of fermions $\psi_r, \psi_r^*$ ($r\in \bZ+\half$)
on the fermionic Fock space.
Let $\{\psi_r, \psi_r^*\}_{r\in\bZ+\half}$ be a set of generators of a Clifford algebra,
satisfying the anti-commutation relations:
\be
\label{eq-anticomm-psi}
[\psi_r,\psi_s]_+=0,
\qquad
[\psi_r^*,\psi_s^*]_+=0,
\qquad
[\psi_r,\psi_s^*]_+= \delta_{r+s,0}\cdot \Id,
\ee
where the bracket is defined by $[\phi,\psi]_+:=\phi\psi+\psi\phi$.
Then the Clifford algebra acts on the fermionic Fock space $\cF$ by:
\ben
\psi_r |\bm a\rangle := z^r \wedge |\bm a\rangle,
\qquad \forall r\in \bZ+\half,
\een
and
\ben
\psi_r^* | \bm a\rangle:=
\begin{cases}
(-1)^{k+1} \cdot z^{a_1}\wedge z^{a_2}\wedge \cdots \wedge \widehat{z^{a_k}} \wedge \cdots,
&\text{if $a_k = -r$ for some $k$;}\\
0, &\text{otherwise.}
\end{cases}
\een

It is clear that $\{\psi_r\}$ all have charge $-1$,
and $\{\psi_r^*\}$ all have charge $1$.
The operators $\{\psi_r,\psi_r^*\}_{r<0}$ are called the fermionic creators,
and $\{\psi_r,\psi_r^*\}_{r>0}$ are called the fermionic annihilators.
One can check that:
\be
\psi_r |0\rangle=0,
\qquad
\psi_r^* |0\rangle=0,
\qquad \forall r>0.
\ee
Moreover,
every element of the form $|\mu\rangle$ (where $\mu$ is a partition)
can be obtained by applying fermionic creators to the vacuum $|0\rangle$ in the following way:
\be
|\mu\rangle=(-1)^{n_1+\cdots+n_k}\cdot
\psi_{-m_1-\half} \psi_{-n_1-\half}^* \cdots
\psi_{-m_k-\half} \psi_{-n_k-\half}^* |0\rangle,
\ee
if the Frobenius notation for $\mu=(\mu_1,\mu_2,\cdots)$ is
$\mu=(m_1,\cdots, m_k | n_1,\cdots,n_k)$.

Furthermore,
one can define an inner product $(\cdot,\cdot)$ on the Fock space $\cF$ by taking
$\{|\bm a\rangle | \text{$\bm a$ is admissible}\}$ to be an orthonormal basis.
Given two admissible sequences $\bm a$ and $\bm b$,
we denote by $\langle \bm b | \bm a\rangle:=(|\bm a\rangle,|\bm b\rangle)$
the inner product of $|\bm a\rangle$ and $|\bm b\rangle$.
It is easy to see that
$\psi_r$ and $\psi_{-r}^*$ are adjoint to each other with respect to this inner product.

\subsection{Affine coordinates on the big cell of the Sato Grassmannian}
\label{sec-pre-affinecoord}

In this subsection we recall the construction of the big cell $Gr_{(0)}$ of the Sato Grassmannian
and the affine coordinates on it.
Then we recall how to associate a vector in the fermionic Fock space $\cF^{(0)}\subset \cF$ to each element of $Gr_{(0)}$
by Bogoliubov transformations.

Let $H$ be the infinite-dimensional vector space:
\ben
H:=\big\{\text{formal series }\sum_{n\in \bZ} a_n z^{n-\half}
\big| a_n=0 \text{ for } n>>0
\big\},
\een
Then $H=H_+ \oplus H_-$,
where
$H_+ = z^\half \bC [z]$ and
$H_- = z^{-\half} \bC [[z^{-1}]]$.
Denote by $\pi_\pm :H\to H_\pm$ the natural projections.
The big cell $Gr_{(0)}$ of the Sato Grassmannian consists of linear subspaces $U\subset H$
such that the $\pi_+: U \to H_+$ is an isomorphism.

Given an element $U\in Gr_{(0)}$,
a basis of the form:
\be
\big\{
\tilde f_n
=z^{n+\half} + \sum_{j<n}\tilde a_{n,j} z^{j+\half}
\big\}_{n\geq 0}
\ee
is called an admissible basis.
Among various admissible bases for $U$,
there is a unique one of the following form
(see eg. \cite{by} or \cite[\S 3]{zhou1}):
\be
\big\{
f_n
=z^{n+\half} + \sum_{m\geq 0} a_{n,m} z^{-m-\half}
\big\}_{n\geq 0},
\ee
called the normalized basis
(also called the canonical basis in some literatures,
see eg. \cite{kmmmz})
for $U$.
The coefficients $\{a_{n,m}\}_{n,m\geq 0}$ of this normalized basis
are called the affine coordinates of $U$ on $Gr_{(0)}$.
This basis determines a semi-infinite wedge product:
\begin{equation*}
\begin{split}
f_0\wedge f_1\wedge f_2\wedge \cdots
:=& \sum \alpha_{m_1,\cdots, m_k;n_1,\cdots, n_k}
z^{-m_1-\half}\wedge\cdots\wedge z^{-m_k-\half} \\
&\quad \wedge z^{\half}\wedge z^{\frac{3}{2}}\wedge\cdots\wedge \widehat{z^{n_k+\half}}
\wedge\cdots \wedge \widehat{z^{n_1+\half}} \wedge \cdots,
\end{split}
\end{equation*}
where $m_1>m_2>\cdots >m_k\geq 0$ and $n_1>n_2>\cdots >n_k\geq 0$ are two sequences of integers,
and the coefficients are given by:
\ben
\alpha_{m_1,\cdots, m_k;n_1,\cdots, n_k}
=(-1)^{n_1+\cdots+n_k} \cdot
\det
\begin{pmatrix}
a_{n_1,m_1} & \cdots & a_{n_1,m_k} \\
\vdots & \vdots & \vdots\\
a_{n_k,m_1} & \cdots & a_{n_k,m_k}
\end{pmatrix}.
\een
Thus we get a linear map:
\be
Gr_{(0)}\to \cF^{(0)},\qquad
U=\text{span}\{f_0,f_1,f_2\cdots\}
\mapsto |U\rangle := f_0\wedge f_1 \wedge f_2 \wedge\cdots,
\ee
where $\{f_n\}_{n\geq 0}$ is the normalized basis for $U$.
It is known by experts that
the Grassmannian approach is in fact equivalent to the fermionic picture,
and the associated element $|U\rangle$
in the fermionic Fock space can be constructed in terms of the
coefficients of an admissible basis of $U\in Gr_{(0)}$ directly,
see eg. \cite[\S 2.5]{kmmmz}.
In this work we will use the following formulation (see \cite[\S 3.6]{zhou1}):
\begin{Theorem}
\label{thm-coeff-Bogoliubov}
Let $\{a_{n,m}\}_{n,m\geq 0}$ be the affine coordinates of $U\in Gr_{(0)}$.
Then $|U\rangle \in \cF^{(0)}$ can be represented as the following Bogoliubov transform:
\ben
|U\rangle = e^A |0\rangle,
\een
where $A:\cF^{(0)}\to\cF^{(0)}$ is defined using the fermionic creators as follows:
\be
\label{eq-def-bog}
A:= \sum_{n,m\geq 0} a_{n,m} \psi_{-m-\half} \psi_{-n-\half}^*.
\ee
\end{Theorem}

The following property will be useful in practical computations.
Let $A$ be defined by \eqref{eq-def-bog},
then for every $r\in\bZ+\half$ one has:
\be
\label{eq-conjugate-1}
e^{-A} \psi_r e^A = \begin{cases}
\psi_r, & r<0;\\
\psi_r - \sum_{m\geq 0} a_{r-\half,m} \psi_{-m-\half},
& r>0,
\end{cases}
\ee
and:
\be
\label{eq-conjugate-2}
e^{-A} \psi_r^* e^A = \begin{cases}
\psi_r^* , & r<0;\\
\psi_r^* + \sum_{n\geq 0} a_{n,r-\half} \psi_{-n-\half}^*,
& r>0.
\end{cases}
\ee
This property can be checked by using:
\ben
e^{-A}\phi e^A =
\phi- [A,\phi] +\frac{1}{2!} [A,[A,\phi]]
-\frac{1}{3!} [A,[A,[A,\phi]]] +\cdots,
\een
where $\phi$ denotes either $\psi_r$ or $\psi_r^*$.
One can compute the right-hand side term by term,
and get:
\begin{equation*}
[A,\psi_r]=\sum_{m,n\geq 0}a_{n,m}\psi_{-m-\half}[\psi_r,\psi_{-n-\half}^*]_+,
\end{equation*}
and thus:
\ben
[A,\psi_r]= \begin{cases}
\sum_{m\geq 0} a_{r-\half,m} \psi_{-m-\half},
& r>0,\\
0, & r<0.
\end{cases}
\een
Then one easily sees that $[A,[A,\psi_r]]=0$ for every $r\in\bZ+\half$,
and thus \eqref{eq-conjugate-1} follows.
And \eqref{eq-conjugate-2} can be similarly proved.

Now let $|U\rangle\in\cF^{(0)}$ be a state in the fermionic Fock space
constructed from an element $U\in Gr_{(0)}$.
Then using the properties \eqref{eq-conjugate-1} and \eqref{eq-conjugate-2},
one can prove that the following Hirota bilinear relation holds:
\be
\label{eq-Hirota-fermion}
\sum_{r\in \bZ+\half}
\psi_r^* |U\rangle \otimes \psi_{-r} |U\rangle =0.
\ee
This is the Pl\"ucker relation on the Sato Grassmannian
(\cite[\S 10]{djm}).

\subsection{The boson-fermion correspondence}
\label{sec-b-f}

We recall the bosonic Fock space and boson-fermion correspondence
in this subsection. See \cite{djm} for details.

Define $\alpha_n$ to be the following operators on the fermionic Fock space $\cF$:
\ben
\alpha_n = \sum_{r\in \bZ+\half} :\psi_{-r} \psi_{r+n}^*:,
\qquad n\in \bZ,
\een
where $: \psi_{-r} \psi_{r+n}^* :$ is the normal-ordered product of fermions
defined by:
\be
:\phi_{r_1}\phi_{r_2}\cdots \phi_{r_n}:
=(-1)^\sigma \phi_{r_{\sigma(1)}}\phi_{r_{\sigma(2)}}\cdots\phi_{r_{\sigma(3)}},
\ee
where $\phi_k$ denotes either $\psi_k$ or $\psi_k^*$ for simplicity,
and $\sigma\in S_n$ is a permutation such that $r_{\sigma(1)}\leq\cdots\leq r_{\sigma(n)}$.
For $n=0$,
the operator $\alpha_0$ is the charge operator on $\cF$.

The operators $\{\alpha_n\}_{n\in \bZ}$ satisfy the following commutation relations:
\ben
[\alpha_m,\alpha_n]= m\delta_{m+n,0} \cdot \Id,
\een
i.e.,
they generate a Heisenberg algebra.
One can also check that:
\ben
[\alpha_n,\psi_r] = \psi_{n+r},
\qquad
[\alpha_n,\psi_r^*] = -\psi_{n+r}^*.
\een

The normal-ordered products for the bosons $\{\alpha_n\}_{n\in \bZ}$
are defined by:
\be
:\alpha_{n_1}\cdots\alpha_{n_k}:=
\alpha_{n_{\sigma(1)}}\cdots\alpha_{n_{\sigma(k)}},
\ee
where $\sigma\in S_k$ such that $n_{\sigma(1)}\leq\cdots \leq n_{\sigma(k)}$.
Then the above commutation relations and the anti-commutation relations \eqref{eq-anticomm-psi} are
equivalent to the following operator product expansions
(see eg. \cite{dms, ka} for some basics of OPEs):
\ben
&&\alpha(\xi)\alpha(\eta)= :\alpha(\xi)\alpha(\eta):+
\frac{1}{(\xi-\eta)^2},\\
&&\alpha(\xi)\psi(\eta)=:\psi(\xi)\psi^*(\xi)\psi(\eta):
+\frac{\psi(\xi)}{\xi-\eta},\\
&&\alpha(\xi)\psi^*(\eta)=:\psi(\xi)\psi^*(\xi)\psi^*(\eta):
-\frac{\psi^*(\xi)}{\xi-\eta},\\
&&\psi(\xi)\psi(\eta)= :\psi(\xi)\psi(\eta):,\\
&&\psi^*(\xi)\psi^*(\eta)= :\psi^*(\xi)\psi^*(\eta):,\\
&&\psi(\xi)\psi^*(\eta)= :\psi(\xi)\psi^*(\eta):+\frac{1}{\xi-\eta},
\een
where
\be
\label{eq-gen-fermi}
\psi(\xi)= \sum_{r\in\bZ+\half} \psi_r \xi^{-r-\half},
\qquad
\psi^*(\xi)= \sum_{r\in\bZ+\half} \psi_r^* \xi^{-r-\half},
\ee
and
\be
\alpha(\xi)= :\psi(\xi)\psi^*(\xi):=
\sum_{n\in\bZ} \alpha_n \xi^{-n-1},
\ee
are the generating series of the fermions and bosons respectively.

The bosonic Fock space $\cB$ is defined to be
$\cB:=\Lambda[w,w^{-1}]$,
where $\Lambda$ is the space of symmetric functions
(for a sequence of formal variables $\bm y=(y_1,y_2\cdots)$,
see \S \ref{sec-symm-fn}),
and $w$ is a formal variable.
In what follows,
we will always omit the variables of the symmetric functions.
The main property of $\Lambda$ that matters for our purpose is that
there are two bases $\{p_\lambda\}$ and $\{s_\lambda\}$ labeled by partitions
satisfying $\deg(p_\lambda)=\deg (s_\lambda) =|\lambda|$ and the relations \eqref{eq-newton-schur}.

The boson-fermion correspondence is a linear isomorphism $\Phi:\cF \to \cB$ of vector spaces,
given by:
\ben
\Phi:\quad
|\bm a\rangle \in \cF^{(m)}
\quad\mapsto\quad
w^m\cdot \langle m |
e^{\sum_{n=1}^\infty \frac{p_n}{n} \alpha_n}
| \bm a \rangle,
\een
where $p_n$ ($n\geq 1$) is the Newton symmetric function of degree $n$.
In particular,
by restricting to $\cF^{(0)}$ we get an isomorphism
(\cite[\S 5]{djm}):
\be
\label{eq-b-f-corresp}
\cF^{(0)}\to \Lambda,
\qquad
|\mu\rangle \mapsto s_\mu =
\langle 0 | e^{\sum_{n=1}^\infty \frac{p_n}{n} \alpha_n} | \mu \rangle,
\ee
where $s_\mu$ is the Schur function labeled by the partition $\mu$.

Using the above isomorphism,
one can rewrite the bosons $\alpha_n$ and fermions $\psi_r,\psi_r^*$ as operators on the bosonic Fock space.
The results are as follows:
\be
\label{eq-boson}
\Phi (\alpha_n | \bm a\rangle)
=\begin{cases}
n\frac{\pd}{\pd p_n} \Phi(|\bm a\rangle), & n>0;\\
p_{-n}\cdot \Phi(|\bm a\rangle), & n<0,
\end{cases}
\ee
and:
\begin{equation*}
\Phi(\psi(\xi)|\bm a\rangle)=
\Psi(\xi) \Phi(|\bm a\rangle),
\qquad\quad
\Phi(\psi^*(\xi)|\bm a\rangle)=
\Psi^*(\xi) \Phi(|\bm a\rangle),
\end{equation*}
where $\Psi(\xi),\Psi^*(\xi)$ are the vertex operators:
\ben
&&\Psi(\xi)=
\exp\bigg( \sum_{n=1}^\infty \frac{p_n}{n} \xi^n \bigg)
\exp\bigg( -\sum_{n=1}^\infty \xi^{-n}\frac{\pd}{\pd p_n} \bigg)
e^K \xi^{\alpha_0},\\
&&\Psi^*(\xi)=
\exp\bigg( -\sum_{n=1}^\infty \frac{p_n}{n} \xi^n \bigg)
\exp\bigg( \sum_{n=1}^\infty \xi^{-n}\frac{\pd}{\pd p_n} \bigg)
e^{-K} \xi^{-\alpha_0},
\een
and the actions of $e^K$ and $\xi^{\alpha_0}$ are defined by:
\begin{equation*}
(e^K f) (z,T)=z\cdot f(z,T),
\qquad\quad
(\xi^{\alpha_0} f)(z,T) = f(\xi z,T).
\end{equation*}

\subsection{Sato's construction of tau-functions}

In this subsection we recall Sato's theory \cite{sa} of the KP hierarchy.
See also Segal-Wilson \cite{sw} for an analytic construction.

In Sato's theory,
the space of all (formal power series) tau-functions is the same as
a semi-infinite Grassmannian $Gr$,
which is an orbit of the trivial solution $\tau=1$ under the action of
the infinite-dimensional Lie group $\widehat{GL}(\infty)$.
In particular, given an element $U\in Gr_{(0)}\subset Gr$,
the function
\ben
\tau_U (\bm x) := \langle 0| e^{\sum_{n\geq 1} x_n \alpha_n} |U\rangle
\een
constructed above is a (formal power series) tau-function of the KP hierarchy
with respect to time variables $\bm x=(x_1,x_2,\cdots)$.

Moreover,
the bilinear relation \eqref{eq-Hirota-fermion} is equivalent to the following
Hirota bilinear equation for the tau-function
(see \cite{djm}):
\be
\label{eq-KP-bilinear2}
\begin{split}
0= \oint\frac{dz}{2\pi i}
\exp\bigg(\sum_{j\geq 1}(x_j-y_j)z^j\bigg)&
\tau(x_1-\frac{1}{z},x_2-\frac{1}{2z^2},x_3-\frac{1}{3z^3} \cdots)\\
&\cdot\tau(y_1+\frac{1}{z},y_2+\frac{1}{2z^2},y_3+\frac{1}{3z^3} \cdots),
\end{split}
\ee
for every $\bm x=(x_1,x_2,\cdots)$ and $\bm y=(y_1,y_2,\cdots)$.
A function $\tau=\tau(\bm x)$
is a tau-function of the KP hierarchy if and only if it satisfies
the above identity.

Since $\cF^{(0)}$ has a basis $\{|\mu\rangle\}$ labeled by partitions,
there exist coefficients $c_\mu(U)$ such that
$|U\rangle = \sum_\mu c_\mu(U) \cdot |\mu\rangle$,
and then the tau-function can be represented as a summation of Schur functions:
\be
\label{eq-Schurexpansion}
\tau_U(\bm x)=\sum_\mu c_\mu(U) s_\mu,
\ee
where $x_n = \frac{p_n}{n}$.
The coefficients $c_\mu(U)$ can be computed using Theorem \ref{thm-coeff-Bogoliubov}.

The wave-function $w$ and dual wave-function $w^*$ associated to the tau-function $\tau_U$ are defined by:
\be
\label{eq-b-f-wave}
\begin{split}
&w(\bm x;z):= \langle -1|
e^{\sum_{n\geq 1} x_n\alpha_n}
\psi(z) |U\rangle /\tau_U,\\
&w^*(\bm x;z):= \langle 1|
e^{\sum_{n\geq 1} x_n\alpha_n}
\psi^*(z) |U\rangle /\tau_U,
\end{split}
\ee
respectively.
They are related to the tau-function by Sato's formulas:
\be
\label{eq-def-wave}
\begin{split}
&w(\bm x;z)
=\exp\bigg( \sum_{n\geq 1} x_n z^n \bigg)
\frac{\tau_U (\bm x -[z^{-1}])}{\tau_U (\bm x)},\\
&w^*(\bm x;z)
=\exp\bigg( -\sum_{n\geq 1} x_n z^n \bigg)
\frac{\tau_U (\bm x +[z^{-1}])}{\tau_U (\bm x)}.
\end{split}
\ee
where $[z^{-1}]$ denotes the sequence $(z^{-1},\frac{z^{-2}}{2},\frac{z^{-3}}{3},\cdots)$.
Then the Hirota bilinear equation \eqref{eq-KP-bilinear2} is equivalent to:
\ben
\Res_{z=\infty} w(\bm x;z) w^*(\bm x';z) =0,
\een
for arbitrary $\bm x$ and $\bm x'$.

The above construction can also be reversed.
If $\tau(\bm x)$ is a tau-function of the KP hierarchy,
then there is an element $U\in Gr_{(0)}$ in the big cell of the Sato Grassmannian
such that $\tau=\tau_U$.
The way to find $U$ is the following
(see \cite{am}):
\be
\label{eq-tau-to-Gr}
U= \text{span}\big\{z^{\half} w^* (0;z),
z^{\half} \pd_x w^* (0;z),
z^{\half} \pd_x^2 w^* (0;z),
\cdots
\big\}\subset H,
\ee
where $x:=x_1$ and $w^*(\bm x;z)$ is the dual wave function determined by \eqref{eq-def-wave}.

Notice that this result has been presented in different notations by different authors,
and in some literatures one uses the wave function $w(\bm x;z)$
instead of the dual wave function $w^* (\bm x;z)$.
Here we describe both results using wave and dual wave functions
in the notations we picked above
(for a proof, see eg. the computations in \cite[\S 4.5]{zhou1}):
\begin{Lemma}
\label{lem-tau-to-Gr}
Let $\{a_{n,m}\}_{n,m\geq 0}$ be the affine coordinates of $U\in Gr_{(0)}$,
and let $w(\bm x;z)$ and $w^*(\bm x;z)$ be the wave function and dual wave function
associated to the tau-function $\tau_U(\bm x)$ respectively.
Then we have:
\begin{equation*}
\begin{split}
& \text{span} \big\{
 w (0;z),
 \pd_x w (0;z),
 \pd_x^2 w (0;z),
\cdots
\big\}=
\text{span} \big\{
z^k-\sum_{n=0}^\infty a_{n,k} z^{-n-1}
\big\}_{k\geq 0},\\
& \text{span} \big\{
 w^* (0;z),
 \pd_x w^* (0;z),
 \pd_x^2 w^* (0;z),
\cdots
\big\}=
\text{span} \big\{
z^k+\sum_{n=0}^\infty a_{k,n} z^{-n-1}
\big\}_{k\geq 0}.
\end{split}
\end{equation*}
\end{Lemma}

\subsection{Zhou's formula for the connected $n$-point functions}

Let $F_U:=\log (\tau_U)$ be the free energy associated to
the tau-function $\tau_U$ where $U\in Gr_{(0)}$.
Define the (all-genera) connected $n$-point function of $F_U$ to be:
\be
G^{(n)}(z_1,\cdots,z_n)
:=\sum_{j_1,\cdots,j_n\geq 1}
\frac{\pd^n F_U (\bm x)}{\pd x_{j_1}\cdots \pd x_{j_n}}\bigg|_{\bm x=0}
\cdot z_1^{-j_1-1} \cdots z_n^{-j_n-1}.
\ee
In \cite[\S 5]{zhou1},
Zhou derived the following formula to compute these $n$-point functions:
\begin{Theorem}
[\cite{zhou1}]
Let $\{a_{n,m}\}_{n,m\geq 0}$ be the affine coordinates for $U$,
and denote:
\be
A(\xi,\eta):= \sum_{m,n\geq 0}
a_{n,m} \xi^{-n-1} \eta^{-m-1}.
\ee
Then:
\be
\label{eq-thm-npt}
G^{(n)}(z_1,\cdots,z_n)
=(-1)^{n-1} \sum_{n\text{-cycles }\sigma}
\prod_{i=1}^{n} \widehat A (z_{\sigma(i)}, z_{\sigma(i+1)})
-\frac{\delta_{n,2}}{(z_1-z_2)^2},
\ee
where $\sigma(n+1):=\sigma(1)$, and:
\be
\widehat A (z_i,z_j)=\begin{cases}
i_{z_i,z_j}\frac{1}{z_i-z_j} + A(z_i,z_j), & i<j;\\
A(z_i,z_i), & i=j;\\
i_{z_j,z_i}\frac{1}{z_i-z_j} + A(z_i,z_j), & i>j,
\end{cases}
\ee
and $i_{\xi,\eta}\frac{1}{\xi+\eta}:= \sum_{k\geq 0} (-1)^k \xi^{-1-k}\eta^k$.
\end{Theorem}

\begin{Remark}
In the framework of emergent geometry developed by Zhou,
it would be better to multiply the $n$-point functions $G^{(n)}(z_1,\cdots,z_n)$ by $dz_1\cdots dz_n$
and then understand them as some multi-linear differentials on the spectral curve
associated to the tau-function $\tau_U$,
where $z$ is a local coordinate on this spectral curve.
See \cite{zhou1} for the emergence of the Airy curve from the Witten-Kontsevich tau-function $\tau^{\text{WK}}$,
and see \cite{zhou8, zhou9} for the emergence of the spectral curves for
the Hermitian one-matrix models and the dessin d'enfants respectively.

From this point of view,
the additional term $\frac{1}{(z_1-z_2)^2}$ in the formula \eqref{eq-thm-npt} for $n=2$
becomes the fundamental bidifferential $\frac{dz_1 dz_2}{(z_1-z_2)^2}$
(which is called the Bergman kernel in the Eynard-Orantin topological recursion \cite{eo}).
One may naturally expect Zhou's formula may be related to the E-O topological recursion
in some interesting manner.
See \cite{zhou10, zhou8} for the emergence of the E-O topological recursion from the Virasoro constraints in
the cases of Hermitian one-matrix models with even couplings and the dessin d'enfants respectively.
\end{Remark}

\subsection{Affine coordinates of the Witten-Kontsevich tau-function}

The affine coordinates $\{a_{n,m}^{\text{WK}}\}_{n,m\geq 0}$ for the Witten-Kontsevich tau-function
$\tau^{\text{WK}}$ was computed
by Zhou in \cite{zhou4}.
See also \cite[\S 6.7]{zhou1}.
The result is:
\be
\label{eq-WK-coeff}
\begin{split}
&a_{3n,3m-1}^{\text{WK}}=a_{3n+2,3m-3}^{\text{WK}}=
\frac{(-1)^n}{36^{m+n}}\cdot \frac{(6m+1)!!}{(2m+2n)!}
\cdot \prod_{j=0}^{n-1} (m+j)\\
&\qquad\qquad
\cdot \prod_{j=1}^n (2m+2j-1) \cdot\big(
B_n(m)+ \frac{2^n \cdot (6n+1)!!}{(2n)!\cdot (6m+1)}
\big),\\
&a_{3n+1,3m-2}^{\text{WK}}=
\frac{(-1)^{n+1}}{36^{m+n}}\cdot \frac{(6m+1)!!}{(2m+2n)!}
\cdot \prod_{j=0}^{n-1} (m+j)\\
&\qquad\qquad
\cdot \prod_{j=1}^n (2m+2j-1) \cdot\big(
B_n(m)+\frac{2^n\cdot (6n+1)!!}{(2n)!\cdot(6m-1)}
\big),
\end{split}
\ee
where $B_n(m)$ are given by:
\be
B_n(m)=\frac{1}{6}\sum_{j=1}^n 108^j \cdot \frac{2^{n-j}\cdot(6n-6j+1)!!}{(2n-2j)!}
\cdot \frac{(m+n)!}{(m+n-j+1)!}.
\ee
And $a_{n,m}^{\text{WK}}=0$ if $n+m\not\equiv -1 (\text{mod }3)$.

In particular,
one has:
\be
\label{eq-WK-coeff-012}
\begin{split}
&a_{0,3m-1}^{\text{WK}} = \frac{1}{36^m}\cdot \frac{(6m+1)!!}{(2m)!}\cdot \frac{1}{6m+1},\\
&a_{1,3m-2}^{\text{WK}} = -\frac{1}{36^m}\cdot \frac{(6m+1)!!}{(2m)!}\cdot \frac{1}{6m-1},\\
&a_{2,3m-3}^{\text{WK}} = \frac{1}{36^m}\cdot \frac{(6m+1)!!}{(2m)!}\cdot \frac{1}{6m+1},
\end{split}
\ee
which are the coefficients of the Faber-Zagier series
(see \cite{ppz}).

Then one can use \eqref{eq-Schurexpansion} to express $\tau^{\text{WK}}$ in terms of Schur functions:
\ben
\tau^{\text{WK}}= \sum_{\mu} a_{\mu}^{\text{WK}}\cdot s_\mu,
\een
where for a partition $\mu=(m_1,\cdots,m_k|n_1,\cdots,n_k)$,
the coefficient is:
\ben
a_\mu^{\text{WK}}=(-1)^{n_1+\cdots+n_k}\cdot \det (a_{n_i,m_j}^{\text{WK}})_{1\leq i,j \leq k}.
\een
Since $a_{n,m}^{\text{WK}}=0$ unless $n+m\equiv -1(\text{mod } 3)$,
one has
$a_\mu^{\text{WK}}=0$ unless $|\mu| \equiv 0 (\text{mod }3)$.

\section{Open Intersection Numbers and the Affine Coordinates}
\label{sec-affinecoord}

In this section we first recall the construction of the open intersection numbers and
the extended partition function $\tau^o$,
and then introduce our main results about the affine coordinates of $\tau^o$ on the Sato Grassmannian.

\subsection{The open intersection numbers}

First we briefly recall some basic results about the open intersection numbers.
See \cite{pst}.

Let $\Delta\subset \bC$ be the unit open disk and $\overline\Delta$ be its closure,
and let $C$ be a closed Riemann surface.
A holomorphic embedding $\Delta\hookrightarrow C$ is called extendable if it extends
to a holomorphic embedding of an open neighborhood of $\overline\Delta$.
Two extendable embeddings are said to be disjoint if their images are disjoint in $C$.
Let $X$ be a connected open Riemann surface
obtained by removing some disjoint extendable disks on a connected closed Riemann surface.
Its boundary $\pd X$ is the union of the images of the unit circle $S^1 \subset \bC$
under the extendable embeddings.
Given such a surface $(X,\pd X)$,
its double $D(X,\pd X)$ is canonically constructed using Schwarz reflections through the boundary components.
The double genus of $(X, \pd X)$ is defined to be the usual genus of its double $D(X,\pd X)$.
Let $\cM_{g,k,l}$ be the moduli space of open Riemann surfaces of double genus $g$,
with $k$ distinct marked points on the boundary $\pd X$
and $l$ distinct marked points in the interior $X\backslash \pd X$.
It is a real orbifold of real dimension $3g-3+k+2l$,
and is nonempty only if the stability condition $3g-3+k+2l>0$ holds.

The open intersection numbers defined in \cite{pst} are:
\be
\langle \tau_{a_1}\tau_{a_2}\cdots\tau_{a_l} \sigma^k \rangle_g^o
:=\int_{\Mbar_{g,k,l}} \psi_1^{a_1} \psi_2^{a_2} \cdots \psi_l^{a_l},
\ee
where $\psi_i$ is the first Chern class of the cotangent line bundle at the $i$-th interior marked point,
and $\Mbar_{g,k,l}$ is supposed to be some suitable compactification of the moduli space $\cM_{g,k,l}$.
The well-definedness of this construction has been completed for the case $g=0$ in \cite{pst}.
Then the open potential $F^{o,\text{geom}}$ is defined to be the following generating series
of these open intersection numbers:
\ben
F^{o,\text{geom}} (u;s;t_0,t_1,\cdots):=
\sum_{g\geq 0} u^{g-1}\cdot\big\langle
\exp\big(s\sigma+\sum_{d\geq 0} t_d \tau_d\big)
\big\rangle_g^o,
\een
where $s$ and $t_0,t_1,t_2,\cdots$ are the coupling constants.
The authors of \cite{pst} conjectured that the potential $F^{o,\text{geom}}(u;s;t_0,t_1,\cdots)$
satisfies a family of differential equations called the open KdV equations.
In \cite{bu2}, Buryak proved that there exists a unique solution $F^o$ to the open KdV
satisfying a suitable initial condition,
thus the conjecture of \cite{pst} becomes $F^{o,\text{geom}}=F^o$.

\subsection{The extended partition function $\tau^o$ and Virasoro constraints}

In \cite{bu,bu2},
Buryak introduced an extended potential
\ben
F^{o,ext} (u;s_0,s_1,\cdots;t_0,t_1,\cdots)
\een
which is a solution the Burgers-KdV equations,
where $s_0:=s$.
The open potential $F^o$ is recovered from $F^{o,ext}$ by
setting $s_n=0$ for every $n\geq 1$.
Buryak suggested that this is the way to include the descendants on the boundary.
Moreover, let
\ben
F^c (u;t_0,t_1,\cdots):=
\sum_{g\geq 0} u^{2g-2}\cdot\big\langle
\exp\big(\sum_{d\geq 0} t_d \tau_d\big)
\big\rangle_g^c
\een
be the generating series of the intersection numbers
\ben
\langle \tau_{a_1}\tau_{a_2}\cdots\tau_{a_n} \rangle_g^c
:=\int_{\Mbar_{g,n}} \psi_1^{a_1} \psi_2^{a_2} \cdots \psi_n^{a_n},
\een
where $\Mbar_{g,n}$ is the Deligne-Mumford moduli space of stable curves \cite{dm,kn}.
In other words,
$F^c := \log \tau^{\text{WK}}$ where $\tau^{\text{WK}}$ is the Witten-Kontsevich tau-function.
Then Buryak proved that the extended partition function
\be
\tau^o := \exp (F^{o,ext}+F^c)
\ee
satisfies the following Virasoro constraints
(\cite[Theorem 1.2]{bu}):
\be
\cL_n^{ext} (\tau^o) =0,
\qquad n\geq -1,
\ee
where the Virasoro operators are:
\be
\label{eq-Virasoro}
\begin{split}
\cL_n^{ext}=&\delta_{n,-1}\frac{t_0^2}{2u^2}+\delta_{n,0}\frac{1}{16}+
\sum_{i\geq 0} \frac{(2i+2n+1)!!}{2^{n+1}\cdot(2i-1)!!}
(t_i-\delta_{i,1})\frac{\pd}{\pd t_{i+n}}\\
&+\frac{u^2}{2} \sum_{i=0}^{n-1} \frac{(2i+1)!! \cdot (2n-2i-1)!!}{2^{n+1}}
\frac{\pd^2}{\pd t_i \pd t_{n-i-1}}\\
&+\delta_{n,-1}\frac{s_0}{u}+\delta_{n,0}\frac{3}{4}
+\sum_{i\geq 0} \frac{(n+i+1)!}{i!}s_i\frac{\pd}{\pd s_{n+i}}
+\frac{3}{4}(n+1)!u \frac{\pd}{\pd s_{n-1}}.
\end{split}
\ee

In \cite{al2,al3},
Alexandrov derived the following matrix integral formula for the extended partition function $\tau^o$
(after taking $u=1$):
\be
\tau^o ([\Lambda])=
\cC^{-1} \cdot \det(\Lambda)\int
[d\Phi] \exp\bigg( -\text{Tr}
\bigg( \frac{\Phi^3}{3!} -\frac{\Lambda^2\Phi}{2} +\log \Phi\bigg)\bigg),
\ee
where the Miwa transform variables
\ben
T_k = \frac{1}{k} \text{Tr} \Lambda^{-k}
\een
are related to the coupling constants $t_k,s_k$ by:
\be
\label{eq-coordchange}
\begin{split}
&t_n=(2n+1)!!\cdot T_{2n+1},\\
&s_n= 2^{n+1}\cdot (n+1)!\cdot T_{2n+2},
\end{split}
\ee
for every $n\geq 0$.
He also proved that after the above change of variables,
the extended partition function $\tau^o$ is a tau-function of the KP hierarchy
with respect to the time variables $\bm T =(T_1,T_2,T_3,\cdots)$,
and thus it corresponds to a point on the Sato Grassmannian.

In the rest of this paper, we will always assume $u=1$.

\subsection{Affine coordinates of $\tau^o$ on the Sato Grassmannian}

Now we state our main result about the affine coordinates $\{a_{n,m}^o\}_{n,m\geq 0}$
of the tau-function $\tau^o$.
The proof of the following theorem will be given in \S \ref{sec-proof}.
\begin{Theorem}
\label{thm-mainthm}
We have:
\be
\label{eq-a0m-explicit}
a_{0,3q-1}^o=\big(\frac{3}{2}\big)^q \cdot (2q-1)!!\cdot \sum_{j=0}^q
\frac{(6j-1)!!}{54^j\cdot (2j)!\cdot (2j-1)!!}.
\ee
Here we use the convention $(-1)!!:=1$.
And for every $p\geq 0$ and $q\geq 1$,
\be
\label{eq-mainthm-relation}
\begin{split}
& a_{3p+1,3q-2}^o = a_{3p,3q-1}^{\text{WK}},\\
& a_{3p+2,3q-3}^o = a_{3p+1,3q-2}^{\text{WK}},\\
& a_{3p+3,3q-1}^o = a_{3p+2,3q}^{\text{WK}} - a_{0,3q-1}^o \cdot a_{3p+2,0}^{\text{WK}},\\
\end{split}
\ee
where $\{a_{n,m}^{\text{WK}}\}_{n,m\geq 0}$ are the affine coordinates of the Witten-Kontsevich tau-function
given by \eqref{eq-WK-coeff}.
Moreover,
we have $a_{n,m}^o=0$ for $n+m\not\equiv -1 (\text{mod }3)$.
\end{Theorem}

\begin{Remark}
The numbers $a_{0,3q-1}^o$ coincide with the numbers $d_q$ given by \cite[(1.9)]{bu}.
\end{Remark}

Now denote by $U^o\in Gr_{(0)}$ the corresponding element in the big cell of the Sato Grassmannian,
then $U^o\subset H$ is spanned by the following normalized basis:
\be
\big\{f_n^o := z^{n+\half} +\sum_{m\geq 0}a_{n,m}^o z^{-m-\half}\big\}_{n\geq 0}.
\ee

\begin{Example}
\label{eg-open-anm}
Here are some examples of the coefficients:
\begin{flalign*}
\begin{split}
f_0^o=& z^{\half}
+\frac{41}{24}z^{-\frac{5}{2}}+\frac{9241}{1152}z^{-\frac{11}{2}}+\frac{5075225}{82944}z^{-\frac{17}{2}}
+\frac{5153008945}{7962624}z^{-\frac{23}{2}}\\
&+\frac{1674966309205}{191102976}z^{-\frac{29}{2}}+
\frac{3985569631633205}{27518828544}z^{-\frac{35}{2}}+\cdots,
\end{split}&&
\end{flalign*}
\begin{flalign*}
\begin{split}
f_1^o=& z^{\frac{3}{2}}
+\frac{ 5}{24}z^{-\frac{3}{2}}+\frac{ 385}{1152}z^{-\frac{9}{2}}+\frac{ 85085}{82944}z^{-\frac{15}{2}}
+\frac{ 37182145}{7962624}z^{-\frac{21}{2}}\\
&+\frac{5391411025}{191102976}z^{-\frac{27}{2}}
+\frac{ 5849680962125}{27518828544}z^{-\frac{33}{2}}
+\cdots,
\end{split}&&
\end{flalign*}
\begin{flalign*}
\begin{split}
f_2^o=& z^{\frac{5}{2}}
-\frac{7}{24}z^{-\frac{1}{2}}-\frac{ 455}{1152}z^{-\frac{7}{2}}-\frac{ 95095}{82944}z^{-\frac{13}{2}}
-\frac{ 40415375}{7962624}z^{-\frac{19}{2}}\\
&-\frac{ 5763232475}{191102976}z^{-\frac{25}{2}}
-\frac{ 6183948445675}{27518828544}z^{-\frac{31}{2}}
-\cdots,
\end{split}&&
\end{flalign*}
\begin{flalign*}
\begin{split}
f_3^o=& z^{\frac{7}{2}}
-\frac{25}{1152}z^{-\frac{5}{2}}-\frac{26765}{41472}z^{-\frac{11}{2}}-\frac{21440785}{2654208}z^{-\frac{17}{2}}-\frac{
5093408425}{47775744}z^{-\frac{23}{2}}
-\cdots,
\end{split}&&
\end{flalign*}
\begin{flalign*}
\begin{split}
f_4^o=& z^{\frac{9}{2}}
-\frac{385}{1152}z^{-\frac{3}{2}}-\frac{43505}{41472}z^{-\frac{9}{2}}-\frac{12677665}{2654208}z^{-\frac{15}{2}}-\frac{
1375739365}{47775744}z^{-\frac{21}{2}}
-\cdots,
\end{split}&&
\end{flalign*}
\begin{flalign*}
\begin{split}
f_5^o=& z^{\frac{11}{2}}
+\frac{455}{1152}z^{-\frac{1}{2}}+\frac{ 45955}{41472}z^{-\frac{7}{2}}
+\frac{ 13028015}{2654208}z^{-\frac{13}{2}}+\frac{ 1398371975}{47775744}z^{-\frac{19}{2}}
+\cdots,
\end{split}&&
\end{flalign*}
\begin{flalign*}
\begin{split}
f_6^o=& z^{\frac{13}{2}} -
\frac{39655}{82944}z^{-\frac{5}{2}}-\frac{ 5562095}{2654208}z^{-\frac{11}{2}}-\frac{ 265839035}{31850496}z^{-\frac{17}{2}}
-\cdots,
\end{split}&&
\end{flalign*}
\begin{flalign*}
\begin{split}
f_7^o=& z^{\frac{15}{2}} +
\frac{85085}{82944}z^{-\frac{3}{2}}+\frac{ 12677665}{2654208}z^{-\frac{9}{2}}
+\frac{ 919343425}{31850496}z^{-\frac{15}{2}}
+\cdots,
\end{split}&&
\end{flalign*}
\begin{flalign*}
\begin{split}
&f_8^o= z^{\frac{17}{2}}
-\frac{95095}{82944}z^{-\frac{1}{2}}-\frac{13028015}{2654208}z^{-\frac{7}{2}}
-\frac{929363435}{31850496}z^{-\frac{13}{2}}
-\cdots.
\end{split}&&
\end{flalign*}
\end{Example}

Now denote by $U^{\text{WK}}\in Gr_{(0)}$ the point on the big cell of the Sato Grassmannian
corresponding to the Witten-Kontsevich tau-function $\tau^{\text{WK}}$,
and denote by
\be
\big\{
f_n^{\text{WK}} := z^{n+\half} + \sum_{m\geq 0} a_{n,m}^{\text{WK}} z^{-m-\half}
\big\}_{n\geq 0}
\ee
the normalized basis of $U^{\text{WK}}$.
Then by \eqref{eq-mainthm-relation} we easily see:
\be
\label{eq-relation-WK}
f_n^o = \begin{cases}
z\cdot f_{n-1}^{\text{WK}},
& \text{if $n\not\equiv 0 (\text{mod }3)$};\\
z\cdot f_{n-1}^{\text{WK}}- a_{n-1,0}^{\text{WK}} \cdot f_0^o,
& \text{if $n\equiv 0 (\text{mod }3)$},\\
\end{cases}
\ee
for every $n\geq 1$,
and thus:
\begin{Corollary}
We have:
\be
\label{eq-relation-subsp}
U^o = z\cdot U^{ \text{WK} } \oplus \bC\cdot f_0^o.
\ee
where $f_0^o$ is given by:
\be
\label{eq-f0-explicit}
f_0^o =\sum_{m\geq 0}
\big(\frac{3}{2}\big)^q \cdot (2q-1)!!\cdot \sum_{j=0}^q
\frac{(6j-1)!!}{54^j\cdot (2j)!\cdot (2j-1)!!}
\cdot z^{-3m+\half}.
\ee
\end{Corollary}
\begin{Remark}
The affine coordinates $\{a_{n,m}^{\text{WK}}\}$
for the Witten-Kontsevich tau-function satisfy the following symmetry
(see Zhou \cite[(65)]{zhou4}):
\be
\label{eq-symm-WK}
a_{n,m}^{\text{WK}} = (-1)^{n+m}\cdot a_{m,n}^{\text{WK}},
\ee
Using this symmetry and \eqref{eq-WK-coeff-012},
the relation \eqref{eq-relation-WK} for $n\equiv 0 (\text{mod }3)$ becomes:
\be
f_{3p+3}^o =
z\cdot f_{3p+2}^{\text{WK}} - \frac{(-1)^p}{36^{p+1}}\cdot \frac{(6p+5)!!}{(2p+2)!}\cdot f_0^o.
\ee
\end{Remark}

\begin{Remark}
The relation \eqref{eq-relation-subsp} can also be derived using
Alexandrov's construction of the admissible basis \cite[(3.16)]{al3},
while finding the explicit formulas for the coefficients of $f_n^o$ needs more computations.
\end{Remark}

\subsection{Recursions and symmetries for the affine coordinates}
\label{sec-recursion}

The affine coordinates $\{a_{n,m}^o\}_{n,m\geq 0}$ satisfy some recursion relations and symmetry conditions.
Denote $a_{l,m}^o :=0$ for $l<0$ or $m<0$,
then the following two theorems are straightforward corollaries of the Virasoro constraints
(see \S \ref{sec-Virasoro-ferm2} for proofs):
\begin{Theorem}
\label{thm-coeff-rec-1}
For every $n\geq -1$,
we have the following $(2n+3)$-step recursion for the affine coordinates:
\be
\begin{split}
a_{l,2n+3+m}^o - &a_{2n+3+l,m}^o
=(2n+\frac{5}{2}+m) a_{l,2n+m}^o
+(l-\frac{3}{2}) a_{2n+l,m}^o\\
&\qquad
+\sum_{k=0}^{2n+2} a_{k,m}^o a_{l,2n+2-k}^o
-\sum_{k=0}^{2n-1} (2n+\frac{3}{2}-k) a_{k,m}^o a_{l,2n-1-k}^o.
\end{split}
\ee
In particular,
for $n=-1$ we have:
\be
\label{eq-1step-rec}
a_{l,m+1}^o- a_{l+1,m}^o=
(m+\half)a_{l,m-2}^o
+(l-\frac{3}{2})a_{l-2,m}^o
+a_{0,m}^o a_{l,0}^o,
\ee
which is the fermionic reformulation of the string equation $\cL_{-1}^{ext}(\tau^o)=0$.
\end{Theorem}

\begin{Theorem}
\label{thm-coeff-rec-2}
For every $n\geq -1$,
we have the following linear constraints for the affine coordinates:
\be
\sum_{k=0}^{2n+2} a_{k,2n+2-k}^o =
\sum_{k=0}^{2n-1} (2n+\frac{3}{2}-k) a_{k,2n-1-k}^o.
\ee
\end{Theorem}

Moreover,
one can apply the symmetry \eqref{eq-symm-WK} of $\{a_{n,m}^{\text{WK}}\}$
to the relations \eqref{eq-mainthm-relation},
and in this way we easily obtain:
\begin{Proposition}
The affine coordinates $\{a_{n,m}^{o}\}_{n,m\geq 0}$ satisfy:
\be
\begin{split}
& a_{3p+2,3q-3}^o = (-1)^{p+q+1} a_{3q-1,3p}^o,\\
& a_{3p+1,3q-2}^o = (-1)^{p+q+1} (a_{3q,3p-1}^o+ a_{0,3p-1}^o a_{3q-1,0}^{\text{WK}}),\\
& a_{3p+3,3q-1}^o = (-1)^{p+q}a_{3q+1,3p+1}^o -a_{0,3q-1}^o a_{3p+2,0}^{\text{WK}}.
\end{split}
\ee
\end{Proposition}

\section{Computations of $\tau^o$ and $\log(\tau^o)$ Using Affine Coordinates}
\label{sec-applications}

In this section,
we show how to compute the tau-function $\tau^o$ and the free energy $\log(\tau^o)$
using the affine coordinates $\{a_{n,m}^o\}$.

\subsection{A formula for $\tau^o$ as a Bogoliubov transform}
\label{sec-f-bogoliubov}

First we regard $\tau^o$ as a vector $|U^o\rangle$ in the fermionic Fock space.
By Theorem \ref{thm-coeff-Bogoliubov},
we know that:
\begin{Theorem}
Let $a_{n,m}^o$ be the affine coordinates given in Theorem \ref{thm-mainthm},
and $A^o$ be the quadratic operator of fermionic creators defined by:
\be
A^o = \sum_{n,m\geq 0} a_{n,m}^o \psi_{-m-\half} \psi_{-n-\half}^*.
\ee
Then the tau-function $\tau^o$ is given by the following Bogoliubov transform
of the fermionic vacuum $|0\rangle$:
\be
|U^o\rangle = e^{A^o} |0\rangle.
\ee
\end{Theorem}

Expanding this exponential $e^{A^o}$,
one gets:
\be
\label{eq-open-expansion-mu}
\begin{split}
|U^o\rangle =& |0\rangle+
\sum_{\mu}
(-1)^{n_1+\cdots+n_k}\cdot\det(a_{n_i,m_j}^o)_{1\leq i,j\leq k} \cdot |\mu\rangle \\
=& |0\rangle + \sum_{\mu} \det(a_{n_i,m_j}^o)_{1\leq i,j\leq k} \cdot
\psi_{-m_1-\half} \psi_{-n_1-\half}^* \cdots \psi_{-m_k-\half} \psi_{-n_k-\half}^*
|0\rangle,
\end{split}
\ee
where the summation $\sum_\mu$ on the right-hand side is over all partitions of positive integers,
and $\mu=(m_1,\cdots,m_k|m_1,\cdots,n_k)$ is the Frobenius notation of a partition $\mu$.
For example,
the first a few terms are:
\begin{equation*}
\begin{split}
|U^o\rangle =& \bigg( 1+
\frac{41}{24}\psi_{-\frac{5}{2}}\psi_{-\frac{1}{2}}^*
+ \frac{5}{24} \psi_{-\frac{3}{2}}\psi_{-\frac{3}{2}}^*
 -\frac{7}{24} \psi_{-\frac{1}{2}}\psi_{-\frac{5}{2}}^* \\
&+\frac{9241}{1152} \psi_{-\frac{11}{2}}\psi_{-\frac{1}{2}}^*
 + \frac{385}{1152} \psi_{-\frac{9}{2}}\psi_{-\frac{3}{2}}^*
  - \frac{455}{1152} \psi_{-\frac{7}{2}}\psi_{-\frac{5}{2}}^*
-\frac{25}{1152} \psi_{-\frac{5}{2}}\psi_{-\frac{7}{2}}^*\\
& - \frac{385}{1152} \psi_{-\frac{3}{2}}\psi_{-\frac{9}{2}}^*
+ \frac{455}{1152} \psi_{-\frac{1}{2}}\psi_{-\frac{11}{2}}^*
-\frac{205}{576} \psi_{-\frac{5}{2}}\psi_{-\frac{3}{2}}^*\psi_{-\frac{3}{2}}\psi_{-\frac{1}{2}}^*\\
& +\frac{287}{576} \psi_{-\frac{5}{2}}\psi_{-\frac{5}{2}}^*\psi_{-\frac{1}{2}}\psi_{-\frac{1}{2}}^*
 +\frac{35}{576} \psi_{-\frac{3}{2}}\psi_{-\frac{5}{2}}^*\psi_{-\frac{1}{2}}\psi_{-\frac{3}{2}}^*
+\cdots \bigg) |0\rangle.
\end{split}
\end{equation*}

\subsection{A formula for $\tau^o$ as a summation of Schur functions}
\label{sec-b-schur}

Via the boson-fermion correspondence,
the tau-function $\tau^o$ can be regarded as a vector in the bosonic Fock space $\Lambda$,
and thus can be represented in terms of Schur functions.
The following is a straightforward consequence of \eqref{eq-open-expansion-mu}:
\begin{Theorem}
Let $T_n=\frac{p_n}{n}$ where $p_n$ are the Newton symmetric functions,
then the tau-function $\tau^o$ is the following summation of Schur functions:
\be
\tau^o (\bm T) = 1+
\sum_{\mu:\text{ } |\mu|>0}
(-1)^{n_1+\cdots+n_k}\cdot\det(a_{n_i,m_j}^o)_{1\leq i,j\leq k} \cdot s_\mu,
\ee
where $a_{n,m}^o$ are given in Theorem \ref{thm-mainthm},
and $\mu=(m_1,\cdots,m_k|m_1,\cdots,n_k)$ is the Frobenius notation of the partition $\mu$.
\end{Theorem}

Similar to the case of the Witten-Kontsevich tau-function,
since $a_{n,m}^o = 0$ for $n+m\not\equiv -1 (\text{mod }3)$
we must have:
\begin{Proposition}
Denote by $a_\mu^o$ the coefficient of $s_\mu$ in $\tau^o$,
then:
\ben
a_\mu^o=0,
\quad \text{if} \quad |\mu|\not\equiv 0 (\text{mod }3).
\een
\end{Proposition}

Using the data listed in Example \ref{eg-open-anm},
we are able to compute first a few terms in $\tau^o$.
Here we write down all the terms $s_{\mu}$ with $|\mu|\leq 9$:
\begin{equation*}
\begin{split}
\tau^o=& 1+\frac{41}{24}s_{(3)} - \frac{5}{24} s_{(2,1)} -\frac{7}{24} s_{(1^3)}\\
&+\frac{9241}{1152} s_{(6)} - \frac{385}{1152} s_{(5,1)} - \frac{455}{1152} s_{(4,1^2)}
+\frac{25}{1152} s_{(3,1^3)}\\
& - \frac{385}{1152} s_{(2,1^4)}
- \frac{455}{1152} s_{(1^6)}
+\frac{205}{576} s_{(3^2)} +\frac{287}{576} s_{(3,2,1)} - \frac{35}{576} s_{(2^3)}\\
&+ \frac{5075225}{82944} s_{(9)} - \frac{85085}{82944} s_{(8,1)}
+ \frac{46205}{27648} s_{(6,3)} - \frac{15785}{27648} s_{(5,4)}\\
&-\frac{95095}{82944} s_{(7,1^2)} +\frac{64687}{27648} s_{(6,2,1)}
-\frac{2695}{27648} s_{(5,2^2)} + \frac{2275}{27648} s_{(4,3,2)}\\
& -\frac{18655}{27648} s_{(4^2,1)} -\frac{1435}{13824} s_{(3^3)}
+ \frac{26765}{41472} s_{(6,1^3)}
-\frac{175}{27648} s_{(3,2^3)}\\
& - \frac{125}{27648} s_{(3^2,2,1)}
-\frac{43505}{41472} s_{(5,1^4)} + \frac{15785}{27648} s_{(3^2,1^3)} + \frac{2695}{27648} s_{(2^4,1)}\\
&-\frac{45955}{41472} s_{(4,1^5)} + \frac{18655}{27648} s_{(3,2,1^4)}
-\frac{2275}{27648} s_{(2^3,1^3)} - \frac{39655}{82944} s_{(3,1^6)}\\
&-\frac{85085}{82944} s_{(2,1^7)} - \frac{95095}{82944} s_{(1^9)}
+\cdots.
\end{split}
\end{equation*}
(We have omitted the terms whose coefficients are zero.)
In particular,
the coefficient of $s_\mu$ for a hook partition $\mu=(m+1,1^n)=(m|n)$ is simply $(-1)^n\cdot a_{n,m}^o$.

\subsection{Generating series of the affine coordinates}
\label{sec-gen}

In the previous subsections
we have discussed the expressions for the partition function $\tau^o$
in terms of the affine coordinates $\{a_{n,m}^o\}_{n,m\geq 0}$.
It is also a natural question to compute the free energy $\log(\tau^o)$
using $\{a_{n,m}^o\}$.
A way to do this is to apply Zhou's formula \eqref{eq-thm-npt} to compute
the generating series of the connected $n$-point correlators.
To apply Zhou's formula,
we are supposed to compute the generating series $A^o(x,y)$ of the affine coordinates first.
We do this in the present subsection.

Let $a(z),b(z)$ be the Faber-Zagier series:
\be
\label{eq-FZ}
\begin{split}
&a(z)=\sum_{m=0}^\infty \frac{(6m-1)!!}{36^m\cdot (2m)!}z^{-3m},\\
&b(z)=-\sum_{m=0}^\infty \frac{(6m-1)!!}{36^m\cdot (2m)!}\frac{6m+1}{6m-1}z^{-3m+1},\\
\end{split}
\ee
and denote by
\be
\label{eq-series-0}
c(z)=\sum_{m\geq 0}
\big(\frac{3}{2}\big)^m \cdot (2m-1)!!\cdot \sum_{j=0}^m
\frac{(6j-1)!!}{54^j\cdot (2j)!\cdot (2j-1)!!}\cdot
z^{-3m}
\ee
the generating series of $\{a_{0,m}^o\}$ (see \eqref{eq-a0m-explicit}).
Now let
\be
A^o (x,y):= \sum_{n,m\geq 0} a_{n,m}^o x^{-n-1} y^{-m-1}
\ee
be the generating series of all $\{a_{n,m}^o\}$,
then by Theorem \ref{thm-mainthm} we have:
\begin{equation*}
\begin{split}
A^o (x,y)
=& x^{-1}\cdot\sum_{q\geq 1} a_{0,3q-1}^o y^{-3q}
+\sum_{n\geq 0, m\geq 1} a_{n,m}^{\text{WK}} y^{-m} x^{-n-2}\\
&-\bigg( \sum_{q\geq 1} a_{0,3q-1}^o y^{-3q}\bigg)\cdot
\bigg( \sum_{p\geq 0} a_{3p+2,0}^{\text{WK}} x^{-3p-4}\bigg)\\
=&
-\bigg(1+ \sum_{m\geq 0} a_{0,m}^o y^{-m-1}\bigg)
\bigg(
-x^{-1}+ \sum_{n\geq 0} a_{n,0}^{\text{WK}} x^{-n-2}
\bigg)\\
& -x^{-1}
+\sum_{n, m\geq 0} a_{n,m}^{\text{WK}} y^{-m} x^{-n-2}.
\end{split}
\end{equation*}
Using the symmetry \eqref{eq-symm-WK},
one gets:
\begin{equation*}
-x^{-1}+ \sum_{n\geq 0} a_{n,0}^{\text{WK}} x^{-n-2}
=-x^{-1}+ \sum_{n\geq 0}(-1)^n a_{0,n}^{\text{WK}} x^{-n-2}
=-x^{-1}\cdot a(-x),
\end{equation*}
where $a(z)$ is the Faber-Zagier series \eqref{eq-FZ}.
Thus:
\begin{equation*}
\begin{split}
A^o(x,y)=& -c(y)\cdot(-\frac{1}{x})a(-x)-\frac{1}{x}
+\frac{y}{x}\sum_{n, m\geq 0} a_{n,m}^{\text{WK}} y^{-m-1} x^{-n-1}\\
=& \frac{c(y)a(-x)-1}{x}
+\frac{y}{x}\sum_{n, m\geq 0} a_{n,m}^{\text{WK}} y^{-m-1} x^{-n-1}.
\end{split}
\end{equation*}
According to \cite[(282)]{zhou1},
one has:
\be
\label{eq-generating-WK}
\sum_{n, m\geq 0} a_{n,m}^{\text{WK}} x^{-n-1} y^{-m-1}
=\frac{1}{y-x} + \frac{a(y)b(-x)-a(-x)b(y)}{y^2-x^2},
\ee
(the notation $A_{m,n}$ in that work is $a_{n,m}^{\text{WK}}$ here).
Plugging this into the above formula for $A^o(x,y)$,
and in this way we have proved:
\begin{Theorem}
\label{thm-generating}
We have:
\be
\label{eq-generating}
A^o(x,y)= \frac{1}{y-x} +
\frac{c(y)a(-x)}{x}
+\frac{y}{x}\cdot
\frac{a(y)b(-x)-a(-x)b(y)}{y^2-x^2},
\ee
where the series $a(z),b(z),c(z)$ are defined by \eqref{eq-FZ} and \eqref{eq-series-0}.
\end{Theorem}

Now using the above theorem and the explicit formulas for the series $a(z),b(z)$ and $c(z)$,
one can compute the coefficients $\{a_{n,m}^o\}$ effectively with the help of Maple or Mathematica.
For example,
\begin{equation*}
\begin{split}
A^o (x,y) =& -\frac{7}{24x^3y}+ \frac{5}{24x^2y^2} + \frac{41}{24xy^3}
+\frac{455}{1152 x^6 y} -\frac{385}{1152 x^5 y^2}\\
&-\frac{25}{1152x^4y^3}
-\frac{455}{1152x^3 y^4} +\frac{385}{1152 x^2 y^5} + \frac{9241}{1152 x y^6}\\
& - \frac{95095}{82944 x^9y} +\frac{85085}{82944 x^8y^2}
-\frac{39655}{82944 x^7y^3}
 +\frac{45955}{41472 x^6y^4}\\
&-\frac{43505}{41472 x^5y^5} -\frac{26765}{41472 x^4y^6} -\frac{95095}{82944 x^3y^7}
+ \frac{85085}{82944 x^2y^8}\\
&+\frac{5075225}{82944 xy^9}
+\cdots.
\end{split}
\end{equation*}

\begin{Remark}
The series $c(z)$ coincide with the series $D(z)$ introduced in \cite[\S 1.3]{bu} (after taking $u=1$)
and the integral $\Phi_1^o(z)$ defined by \cite[(3.16)]{al3}.
Using the fermionic reformulation \eqref{eq-1step-rec} of the string equation,
one can easily derive the following relations among these series
(see also Buryak \cite[\S A.2]{bu},
and Alexandrov's Kac-Schwarz operator \cite[(3.19)]{al3}):
\ben
&&a(z)=\big(z^{-2}\pd_z +1 -\frac{3}{2} z^{-3}\big) c(z),\\
&&b(z)=\big(z^{-2}\pd_z +1 -\frac{3}{2} z^{-3}\big) \big(z\cdot a(z)\big).
\een
\end{Remark}

\subsection{Computing the $n$-point functions using Zhou's formula}
\label{sec-npt}

Now we can apply Theorem \ref{thm-generating}
to Zhou's formula \eqref{eq-thm-npt} to compute the following $n$-point functions
of $\log(\tau^o)$:
\be
G_{(n)}^o(z_1,\cdots,z_n)
:=\sum_{j_1,\cdots,j_n\geq 1}
\frac{\pd^n \log\tau^o (\bm T)}{\pd T_{j_1}\cdots \pd T_{j_n}}\bigg|_{\bm T=0}
\cdot z_1^{-j_1-1} \cdots z_n^{-j_n-1}.
\ee

Here we give some examples.
For $n=1$,
the formula \eqref{eq-thm-npt} simply gives:
\begin{equation*}
G_{(1)}^o (z)=A^o(z,z)
= \sum_{k\geq 0} \bigg(\sum_{n+m=k} a_{n,m}^o \bigg) z^{-k-2},
\end{equation*}
thus:
\begin{equation*}
\begin{split}
G_{(1)}^o (z)=&
\frac{13}{8 z^4} + \frac{8}{z^7} + \frac{7665}{128 z^{10}} + \frac{640}{z^{13}}
+\frac{8853845}{1024 z^{16}} + \frac{143360}{z^{19}} + \frac{91699252645}{
 32768 z^{22}}\\
 & + \frac{63078400}{z^{25}} + \frac{422005261552875}{
 262144 z^{28}} + \frac{45921075200}{z^{31}}\\
 & + \frac{6070737898677577125}{
 4194304 z^{34}} + \frac{49962129817600}{z^{37}}\\
 & +\frac{62894308834692923713125}{33554432 z^{40}}
 +\frac{75942437322752000}{z^{43}}\\
 & + \frac{7096529123478095160893158125}{ 2147483648 z^{46}}
 +\cdots.
\end{split}
\end{equation*}

For $n=2$, there is only one $2$-cycle $(1,2)$,
thus:
\be
\begin{split}
\label{eq-formula-2pt-x}
G_{(2)}^o (z_1,z_2)=&
-\big(\frac{1}{z_1-z_2}+A^o(z_1,z_2)\big)\big(\frac{1}{z_2-z_1}+A^o(z_2,z_1)\big)
-\frac{1}{(z_1-z_2)^2}\\
=&\frac{A^o(z_1,z_2)-A^o(z_2,z_1)}{z_1-z_2} - A^o(z_1,z_2)A^o(z_2,z_1),
\end{split}
\ee
and then by \eqref{eq-generating} one has:
\begin{equation*}
\begin{split}
&G_{(2)}^o (z_1,z_2)
=\frac{2}{z_1^2 z_2^3}+\frac{2}{z_1^3 z_2^2}
+\frac{65}{8 z_1^2 z_2^6}
+\frac{8}{ z_1^3 z_2^5}
+\frac{39}{8 z_1^4 z_2^4}
+\frac{8}{ z_1^5 z_2^3}
+\frac{65}{8 z_1^6 z_2^2}\\
&\qquad\quad
+\frac{64}{z_1^2z_2^9}
+\frac{245}{4 z_1^3z_2^8} + \frac{48}{z_1^4z_2^7}
+\frac{60}{z_1^5 z_2^6} +\frac{60}{z_1^6 z_2^5}
+\frac{48}{z_1^7z_2^4}
+\frac{245}{4 z_1^8z_2^3} +\frac{64}{z_1^9z_2^2}\\
&\qquad\quad
+\frac{84315}{128z_1^2 z_2^{12}} +\frac{640}{z_1^3z_2^{11}}
+\frac{68985}{128 z_1^4 z_2^{10}}
+\frac{640}{z_1^5 z_2^9} +\frac{80815}{128 z_1^6 z_2^8} +\frac{576}{z_1^7 z_2^7}
 +\frac{80815}{128 z_1^8 z_2^6}\\
 &\qquad\quad
  +\frac{640}{z_1^9 z_2^5} +\frac{68985}{128 z_1^{10} z_2^{4}}
+\frac{640}{z_1^{11}z_2^{3}}
 +\frac{84315}{128z_1^{12} z_2^{2}}
+\frac{8960}{z_1^2z_2^{15}} + \frac{555555}{64 z_1^3 z_2^{14}}
+\frac{7680}{z_1^4 z_2^{13}}\\
&\qquad\quad + \frac{17325}{2 z_1^5 z_2^{12}}
+\frac{8640}{z_1^6 z_2^{11}} + \frac{8190}{z_1^7 z_2^{10}}
+\frac{8680}{z_1^8 z_2^9}
 + \frac{8680}{z_1^9 z_2^8}
+ \frac{8190}{z_1^{10} z_2^{7}}+\frac{8640}{z_1^{11} z_2^{6}}
+ \frac{17325}{2 z_1^{12} z_2^{5}}\\
&\qquad\quad +\frac{7680}{z_1^{13} z_2^{4}}
+ \frac{555555}{64 z_1^{14} z_2^{3}} +\frac{8960}{z_1^{15}z_2^{2}}
+ \frac{150 515 365}{1024 z_1^2 z_2^{18}}
+ \frac{143 360}{z_1^3 z_2^{17}}
+ \frac{132 807 675}{1024 z_1^4 z_2^{16}}\\
&\qquad\quad
+ \frac{143 360}{z_1^5 z_2^{15}}
+ \frac{146 151 005}{1024 z_1^6 z_2^{14}}
+ \frac{138 240}{z_1^7 z_2^{13}}
+ \frac{73 249 715}{512 z_1^8 z_2^{12}}
+ \frac{143 360}{z_1^9 z_2^{11}}
+\frac{71 395 065}{512 z_1^{10} z_2^{10}}\\
&\qquad\quad
+ \frac{143 360}{z_1^{11} z_2^{9}}
+ \frac{73 249 715}{512 z_1^{12} z_2^8}
+ \frac{138 240}{z_1^{13} z_2^{7}}
+ \frac{146 151 005}{1024 z_1^{14} z_2^{6}}
+ \frac{143 360}{z_1^{15} z_2^{5}}
+ \frac{132 807 675}{1024 z_1^{16} z_2^4}\\
&\qquad \quad
+ \frac{143 360}{z_1^{17} z_2^3}
+ \frac{150 515 365}{1024 z_1^{18} z_2^{2}}
+\cdots.
\end{split}
\end{equation*}

Now consider the case $n=3$.
In this case there are two $3$-cycles $(1,2,3)$ and $(1,3,2)$,
therefore the formula \eqref{eq-thm-npt} gives:
\begin{equation*}
\begin{split}
&G_{(3)}^o (z_1,z_2,z_3)\\
= &
\big(\frac{1}{z_1-z_2}+A^o(z_1,z_2)\big)\big(\frac{1}{z_2-z_3}+A^o(z_2,z_3)\big)
\big(\frac{1}{z_3-z_1}+A^o(z_3,z_1)\big)\\
&+\big(\frac{1}{z_1-z_3}+A^o(z_1,z_3)\big)\big(\frac{1}{z_3-z_2}+A^o(z_3,z_2)\big)
\big(\frac{1}{z_2-z_1}+A^o(z_2,z_1)\big)\\
=&A^o(z_1,z_2)A^o(z_2,z_3)A^o(z_3,z_1)+A^o(z_1,z_3)A^o(z_3,z_2)A^o(z_2,z_1)\\
& + \frac{A^o(z_2,z_3)A^o(z_3,z_1)-A^o(z_1,z_3)A^o(z_3,z_2)}{z_1-z_2}\\
& + \frac{A^o(z_3,z_1)A^o(z_1,z_2)-A^o(z_2,z_1)A^o(z_1,z_3)}{z_2-z_3}\\
& + \frac{A^o(z_1,z_2)A^o(z_2,z_3)-A^o(z_3,z_2)A^o(z_2,z_1)}{z_3-z_1}\\
& + \frac{A^o(z_2,z_3)+A^o(z_3,z_2)}{(z_3-z_1)(z_1-z_2)}
+ \frac{A^o(z_1,z_2)+A^o(z_2,z_1)}{(z_2-z_3)(z_3-z_1)}
+ \frac{A^o(z_3,z_1)+A^o(z_1,z_3)}{(z_1-z_2)(z_2-z_3)}.
\end{split}
\end{equation*}
For simplicity we will denote by $\frac{1}{z^{(m,n,l)}}$ the summation of all distinct terms
of the form $\frac{1}{z_i^m z_j^n z_k^l}$ where $\{i,j,k\}=\{1,2,3\}$,
for example:
\begin{equation*}
\begin{split}
&\frac{1}{z^{(2,2,2)}}:=\frac{1}{z_1^2z_2^2z_3^2},\\
&\frac{1}{z^{(2,2,5)}}:=
\frac{1}{z_1^2z_2^2z_3^5}+\frac{1}{z_1^2z_2^5z_3^2}+\frac{1}{z_1^5z_2^2z_3^2},\\
&\frac{1}{z^{(2,3,4)}}:=
\frac{1}{z_1^2z_2^3z_3^4} +\frac{1}{z_1^2z_2^4z_3^3}
+\frac{1}{z_1^3z_2^2z_3^4}
+\frac{1}{z_1^3z_2^4z_3^2} +\frac{1}{z_1^4z_2^2z_3^3} +\frac{1}{z_1^4z_2^3z_3^2}.
\end{split}
\end{equation*}
Then using \eqref{eq-generating} one has:
\begin{equation*}
\begin{split}
&G_{(3)}^o (z_1,z_2,z_3)= \frac{1}{z^{(2,2,2)}}+
\frac{8}{z^{(2,2,5)}}
+\frac{6}{z^{(2,3,4)}}
+\frac{8}{z^{(3,3,3)}}
 +\frac{455}{8 z^{(2,2,8)}}
+\frac{48}{z^{(2,3,7)}}\\
&\qquad\qquad\quad
+\frac{195}{4 z^{(2,4,6)}}
+\frac{64}{z^{(2,5,5)}}
+\frac{60}{z^{(3,3,6)}}
+\frac{48}{z^{(3,4,5)}}
+\frac{117}{4 z^{(4,4,4)}}
 + \frac{640}{z^{(2,2,11)}}\\
&\qquad\qquad\quad
+\frac{2205}{4 z^{(2,3,10)}}
+\frac{576}{z^{(2,4,9)}}
+\frac{665}{z^{(2,5,8)}}
 +\frac{600}{z^{(2,6,7)}}
+\frac{640}{z^{(3,3,9)}}
+\frac{2205}{4 z^{(3,4,8)}}
\\
&\qquad\qquad\quad
+\frac{576}{z^{(3,5,7)}}
+\frac{1275}{2 z^{(3,6,6)}}
+\frac{432}{z^{(4,4,7)}}
+\frac{540}{z^{(4,5,6)}}
+\frac{640}{z^{(5,5,5)}}
+ \frac{1 096 095}{128 z^{(2,2,14)}}\\
&\qquad\qquad\quad
+ \frac{7680}{z^{(2,3,13)}}
+ \frac{252 945}{32 z^{(2,4,12)}}
+ \frac{17 325}{2 z^{(3,3,12)}}
+ \frac{8960}{z^{(2,5,11)}}
+ \frac{7680}{z^{(3,4,11)}}\\
&\qquad\qquad\quad
+ \frac{268 065}{32 z^{(2,6,10)}}
+ \frac{8190}{z^{(3,5,10)}}
+ \frac{206 955}{32 z^{(4,4,10)}}
+ \frac{8448}{z^{(2,7,9)}}
+ \frac{8640}{z^{(3,6,9)}}
+ \frac{7680}{z^{(4,5,9)}}\\
&\qquad\qquad\quad
+ \frac{565 705}{64 z^{(2,8,8)}}
+ \frac{8190}{z^{(3,7,8)}}
+ \frac{242 445}{32 z^{(4,6,8)}}
+ \frac{8680}{z^{(5,5,8)}}
+ \frac{6912}{z^{(4,7,7)}}
+ \frac{8160}{z^{(5,6,7)}}\\
&\qquad\qquad\quad
+ \frac{136 575}{16 z^{(6,6,6)}}
+\cdots.
\end{split}
\end{equation*}

For $n=4$,
there are $6$ different $4$-cycles $(1,2,3,4)$,
$(1,2,4,3)$,
$(1,3,2,4)$,
$(1,3,4,2)$,
$(1,4,2,3)$,
$(1,4,3,2)$.
Similarly we will denote by $\frac{1}{z^{(a,b,c,d)}}$ the summation of all possible distinct terms
of the form $\frac{1}{z_i^a z_j^b z_k^c z_l^d}$ where $\{i,j,k,l\}=\{1,2,3,4\}$.
Then using \eqref{eq-thm-npt} we have:
\begin{equation*}
\begin{split}
&G_{(4)}^o (z_1,z_2,z_3,z_4)=
\frac{3}{z^{(2,2,2,4)}}+
\frac{48}{z^{(2,2,2,7)}}+
\frac{30}{z^{(2,2,3,6)}}+
\frac{48}{z^{(2,2,4,5)}}+
\frac{32}{z^{(2,3,3,5)}}\\
&\qquad\qquad\quad +
\frac{36}{z^{(2,3,4,4)}}+
\frac{48}{z^{(3,3,3,4)}}+
\frac{4095}{8 z^{(2,2,2,10)}}+
\frac{384}{z^{(2,2,3,9)}}+
\frac{4095}{8 z^{(2,2,4,8)}}\\
&\qquad\qquad\quad +
\frac{576}{z^{(2,2,5,7)}}
\frac{975}{2 z^{(2,2,6,6)}}+
\frac{420}{z^{(2,3,3,8)}}+
\frac{432}{z^{(2,3,4,7)}}+
\frac{480}{z^{(2,3,5,6)}}\\
&\qquad\qquad\quad +
\frac{1755}{4 z^{(2,4,4,6)}}+
\frac{576}{z^{(2,4,5,5)}} +
\frac{576}{z^{(3,3,3,7)}}+
\frac{540}{z^{(3,3,4,6)}}+
\frac{512}{z^{(3,3,5,5)}}\\
&\qquad\qquad\quad +
\frac{432}{z^{(3,4,4,5)}}+
\frac{1053}{4 z^{(4,4,4,4)}}+ \cdots.
\end{split}
\end{equation*}
From the above data one can read off some data of the free energy:
\begin{equation*}
\begin{split}
&\log(\tau^o)= \bigg(
\frac{13}{8} T_3 + 8 T_6 + \frac{7665}{128} T_9 + 640 T_{12}
+\frac{8853845}{1024} T_{15} + 143360 T_{18} +\cdots\bigg) \\
&\qquad +\bigg(
2 T_1T_2
+\frac{65}{8} T_1T_5
+8 T_2T_4
+\frac{39}{16} T_3^2
+64 T_1T_8
+\frac{245}{4 }T_2T_7 + 48 T_3T_6\\
&\qquad
+60 T_4T_5
 +\frac{84315}{128} T_1T_{11} +640 T_2T_{10}
  +\frac{68985}{128 } T_3 T_9
+640 T_4T_8\\
&\qquad +\frac{80815}{128 }T_5T_7
 +288 T_6^2+
 +8960T_1T_{14} + \frac{555555}{64 }T_2T_{13}
+7680 T_3T_{12}\\
&\qquad + \frac{17325}{2 } T_4T_{11}
+8640 T_5T_{10} + 8190 T_6T_9
+8680 T_7T_8
+ \cdots\bigg)\\
&\qquad +\bigg(
\frac{1}{6}T_1^3+
4 T_1^2T_4
+6 T_1T_2T_3
+\frac{4}{3} T_2^3
 +\frac{455}{16} T_1^2 T_7
+48 T_1T_2T_6 +\frac{195}{4 } T_1T_3T_5\\
&\qquad
 +30 T_2^2T_5
+32 T_1T_4^2
 +48 T_2T_3T_4
+\frac{39}{8} T_3^3 + 320 T_1^2T_{10}
 +\frac{2205}{4 } T_1T_2T_9\\
&\qquad
+576 T_1T_3T_8
 + 320 T_2^2T_8
 +665 T_1T_4T_7
+\frac{2205}{4 }T_2T_3T_7
 +600 T_1T_5T_6\\
 &\qquad
+576 T_2T_4T_6
+216 T_3^2T_6
 +\frac{1275}{4} T_2 T_5^2
 +540 T_3T_4T_5 +\frac{320}{3} T_4^3 +\cdots
\bigg)\\
&\qquad
+\bigg( \half T_1^3T_3
+ 8 T_1^3T_6
+ 15 T_1^2 T_2T_5
+ 24 T_1^2 T_3T_4
+ 16 T_1 T_2^2 T_4\\
&\qquad + 18 T_1 T_2 T_3^2
+ 8 T_2^3 T_3 +\cdots \bigg)
 + \cdots\cdots,
\end{split}
\end{equation*}
where the time variables $T_i$ are related to the
coupling constants $t_i$ and $s_i$ of the open intersection theory
by \eqref{eq-coordchange}.

\begin{Remark}
In \cite[Theorem 1.6]{br},
Bertola and Ruzza have derived another algorithm to compute the $n$-point functions
for the open intersections numbers and the Kontsevich-Penner models using a different method.
Their formula is similar but somehow different from Zhou's formula \eqref{eq-thm-npt}
for an arbitrary tau-function of KP.
It will be interesting to compare their matrix $A(\lambda)$
(see \cite[(1.8)]{br})
with the generating series $A^o(x,y)$ of the affine coordinates.
\end{Remark}

\subsection{Formulas for the $1$-point functions}
\label{sec-relation-npt-1}

As an application of Zhou's formula,
we write down some closed formulas for the $n$-point functions $G_{(n)}^o (z)$ for small $n$.

First consider the case $n=1$.
Recall that $G_{(1)}^o (z)= A^o(z,z)$
where $A^o(x,y)$ is given by \eqref{eq-generating},
thus by L'Hopital's rule we obtain:
\be
\begin{split}
G_{(1)}^o (z)
=& \lim_{y\to z} \bigg(
\frac{1}{y-z} +
\frac{c(y)a(-z)}{z}
+\frac{y}{z}\cdot
\frac{a(y)b(-z)-a(-z)b(y)}{y^2-z^2}
\bigg)\\
=& \frac{1}{2z} + \frac{c(z)a(-z)}{z} +\frac{a'(z)b(-z)-a(-z)b'(z)}{2z}\\
&+\frac{a(z)b(-z)-a(-z)b(z)}{2z^2}.
\end{split}
\ee
Moreover,
it is known that the Faber-Zagier series $a(z),b(z)$ satisfy the following relation
(see eg. \cite{iz} or \cite{ppz}):
\be
\label{eq-relation-ab}
a(z)b(-z)-a(-z)b(z)=-2z,
\ee
thus the above formula becomes:
\begin{Proposition}
We have:
\be
G_{(1)}^o (z)
=-\frac{1}{2z} + \frac{c(z)a(-z)}{z} +\frac{a'(z)b(-z)-a(-z)b'(z)}{2z}.
\ee
\end{Proposition}

Another way to compute the one-point function $G_{(1)}^o (z)$
is to compare $G_{(1)}^o (z)$ with the one-point function of the Witten-Kontsevich tau-function.
In fact,
denote by
\be
G_{(n)}^{\text{WK}}(z_1,\cdots,z_n)
:=\sum_{j_1,\cdots,j_n\geq 1}
\frac{\pd^n \log\tau^{\text{WK}} (\bm T)}{\pd T_{j_1}\cdots \pd T_{j_n}}\bigg|_{\bm T=0}
\cdot z_1^{-j_1-1} \cdots z_n^{-j_n-1}
\ee
the connected $n$-point function of the Witten-Kontsevich tau-function
(whose coefficients are the intersection numbers of psi-classes on $\Mbar_{g,n}$),
where $\bm T=(T_1,T_2\cdots)$ are the time variables of the KP flows (not the KdV flows!).
Then by Zhou's formula \eqref{eq-thm-npt},
the $1$-point function is given by:
\be
G_{(1)}^{\text{WK}}(z) = A^{\text{WK}}(z,z),
\ee
where (see \eqref{eq-generating-WK}):
\be
\label{eq-generating-WK-2}
\begin{split}
A^{\text{WK}}(x,y):=&
\sum_{n, m\geq 0} a_{n,m}^{\text{WK}} x^{-n-1} y^{-m-1}\\
=&\frac{1}{y-x} + \frac{a(y)b(-x)-a(-x)b(y)}{y^2-x^2}.
\end{split}
\ee
Thus by \eqref{eq-generating} we have:
\be
\label{eq-generating-relation}
\begin{split}
A^o(x,y) - A^{\text{WK}}(x,y) =&
\frac{c(y)a(-x)}{x}
+\big(\frac{y}{x}-1\big)\cdot
\frac{a(y)b(-x)-a(-x)b(y)}{y^2-x^2}\\
=&\frac{c(y)a(-x)}{x}
+ \frac{a(y)b(-x)-a(-x)b(y)}{(y+x)x},
\end{split}
\ee
therefore:
\be
\begin{split}
G_{(1)}^o(z) -
G_{(1)}^{\text{WK}}(z)=&
A^o(z,z)-A^{\text{WK}}(z,z)\\
=&\frac{c(z)a(-z)}{z}
+ \frac{a(z)b(-z)-a(-z)b(z)}{2z^2}.
\end{split}
\ee
Again using \eqref{eq-relation-ab},
the above relation becomes:
\be
\begin{split}
G_{(1)}^o(z) =&
G_{(1)}^{\text{WK}}(z)
+\frac{c(z)a(-z)}{z}+ \frac{(-2z)}{2z^2}\\
=&
G_{(1)}^{\text{WK}}(z)
+\frac{c(z)a(-z)-1}{z}.
\end{split}
\ee
It is known in literatures that the one-point correlators of the Witten-Kontsevich tau-function
are given by (\cite{wi}):
\ben
\langle \tau_{3g-2}\rangle_g =\frac{1}{24^g \cdot g!},
\een
thus (see \cite[(295)]{zhou1}):
\ben
G_{(1)}^{\text{WK}}(z)=
\sum_{g\geq 1} \frac{(6g-3)!!}{24^g \cdot g!} \cdot z^{-6g+2},
\een
thus we obtain the following formula for the one-point function of $\log(\tau^o)$:
\begin{Proposition}
\label{prop-1pt}
We have:
\be
G_{(1)}^o (z)=\frac{c(z)a(-z)-1}{z}
+\sum_{g\geq 1} \frac{(6g-3)!!}{24^g \cdot g!} \cdot z^{-6g+2},
\ee
where $a(z),c(z)$ are given by \eqref{eq-FZ} and \eqref{eq-series-0}.
\end{Proposition}

Recall that the extended partition function is defined by $\tau^o = \exp(F^{o,ext}+F^c)$
where $F^c = \log (\tau^{\text{WK}})$ is the free energy of the Witten-Kontsevich tau-function,
thus the above formula simply tells that:
\begin{Corollary}
We have:
\be
\sum_{j\geq 1}
\frac{\pd F^{o,ext}}{\pd T_j} \bigg|_{\bm T=0}
\cdot z^{-j-1}
=\frac{c(z)a(-z)-1}{z}.
\ee
\end{Corollary}

\subsection{Formulas for the $2$-point functions}
\label{sec-relation-npt-2}

Using \eqref{eq-generating-relation}
one can also write down the relations between $G_{(n)}^o$ and $G_{(n)}^{\text{WK}}$
for general $n$.
However,
the formulas will be complicated for $n\geq 2$.
For example,
denote:
\begin{equation*}
\alpha(x,y):=
A^o(x,y) - A^{\text{WK}}(x,y)
=\frac{c(y)a(-x)}{x}
+ \frac{a(y)b(-x)-a(-x)b(y)}{(y+x)x},
\end{equation*}
and then by \eqref{eq-formula-2pt-x} we have:
\begin{equation*}
\begin{split}
G_{(2)}^o(x,y) - G_{(2)}^{\text{WK}}(x,y)
=&\frac{\alpha(x,y)-\alpha(y,x)}{x-y}
-\alpha(x,y)A^{\text{WK}}(y,x)\\
&-\alpha(y,x)A^{\text{WK}}(x,y)
-\alpha(x,y) \alpha(y,x).
\end{split}
\end{equation*}
Recall that by definition
\ben
G_{(2)}^o(x,y) - G_{(2)}^{\text{WK}}(x,y)
= \sum_{j,k\geq 1}\frac{\pd^2 F^{o,ext}}{\pd T_j\pd T_k}\bigg|_{\bm T =0} \cdot
x^{-j-1} y^{-k-1},
\een
and $A^{\text{WK}}(x,y)$ is given by \eqref{eq-generating-WK-2},
then after some direct computations and simplifications we obtain
the following formula for the $2$-point function:
\begin{Proposition}
We have:
\be
\begin{split}
&\sum_{j,k\geq 1}\frac{\pd^2 F^{o,ext}}{\pd T_j\pd T_k}\bigg|_{\bm T =0} \cdot
x^{-j-1} y^{-k-1}\\
=&
\frac{1}{x(x^2-y^2)}\bigg(
a(y)a(-y)b(-x)c(x) - a(-x)a(-y)b(y)c(x)
\bigg)\\
&+\frac{1}{y(x^2-y^2)}\bigg(
a(-x)a(-y)b(x)c(y) - a(x)a(-x)b(-y)c(y)
\bigg)\\
&- \frac{a(-x)a(-y)c(x)c(y)}{xy},
\end{split}
\ee
where $a(z),b(z),c(z)$ are given by \eqref{eq-FZ} and \eqref{eq-series-0}.
\end{Proposition}

Furthermore,
one can also plug \eqref{eq-generating-WK-2} into \eqref{eq-thm-npt} and obtain:
\be
\begin{split}
G_{(2)}^{\text{WK}}=&
\frac{1}{(x^2-y^2)^2}
\bigg(a(x)a(y)b(-x)b(-y) +a(-x)a(-y)b(x)b(y)\\
&-a(y)a(-y)b(x)b(-x)   - a(x)a(-x)b(y)b(-y)\bigg)
-\frac{1}{(x-y)^2},
\end{split}
\ee
thus the above formula gives us the following:
\begin{Corollary}
We have:
\be
\label{eq-formula-2pt}
\begin{split}
G_{(2)}^o(x,y)=&
\sum_{j,k\geq 1} \frac{\pd^2 \log \tau^o (\bm T)}{\pd T_j \pd T_k}\bigg|_{\bm T =0}
\cdot x^{-j-1}y^{-k-1}\\
=& -\frac{1}{(x-y)^2}
- \frac{a(-x)a(-y)c(x)c(y)}{xy}\\
&+\frac{1}{x(x^2-y^2)}\bigg(
a(y)a(-y)b(-x)c(x) - a(-x)a(-y)b(y)c(x)
\bigg)\\
&+\frac{1}{y(x^2-y^2)}\bigg(
a(-x)a(-y)b(x)c(y) - a(x)a(-x)b(-y)c(y)
\bigg)\\
&+\frac{1}{(x^2-y^2)^2}
\bigg(a(x)a(y)b(-x)b(-y) +a(-x)a(-y)b(x)b(y)\\
&\qquad\qquad
-a(y)a(-y)b(x)b(-x)   - a(x)a(-x)b(y)b(-y)\bigg).
\end{split}
\ee
\end{Corollary}

\subsection{Formulas for the $3$-point functions}
\label{sec-relation-npt-3}

Now we write down the formulas for the $3$-point functions.
Here we omit the details since they are similar to the computations
in the last two subsections.
By \eqref{eq-thm-npt} and \eqref{eq-generating} we have:
\begin{equation}
\begin{split}
&G_{(3)}^o(x,y,z) =  \bigg[
 \frac{a(-x)a(-y)a(-z)b(x)b(y)c(z)}{2z (z^2-x^2)(z^2-y^2)}\\
&\qquad+ \frac{(x^2+y^2-2z^2)a(-x)a(-y)a(-z)b(z)c(x)c(y)}{2xy (z^2-x^2)(z^2-y^2)}\\
&\qquad+\frac{a(x)a(-x)b(y)b(-y)a(-z)b(z)}{(x^2-y^2)(y^2-z^2)(z^2-x^2)}
+\frac{a(x)a(-x)b(y)b(-y)a(z)b(-z)}{(x^2-y^2)(y^2-z^2)(z^2-x^2)}\\
&\qquad+\frac{a(x)a(-x)b(-z)c(z)a(y)b(-y)}{z(x^2-y^2)(y^2-z^2)}
+\frac{a(x)a(-x)b(-z)c(z)a(-y)b(y)}{z(x^2-y^2)(x^2-z^2)}\\
&\qquad+\frac{a(x)a(-x)c(y)c(z)a(-y)b(-z)}{yz(x^2-z^2)}
+\frac{a(x)a(-x)b(y)b(-y)a(-z)c(z)}{z(y^2-x^2)(y^2-z^2)}\\
&\qquad+\frac{a(-x)a(-y)a(-z)c(x)c(y)c(z)}{3xyz}\bigg]
+\sum_{\sigma\in S_3\backslash\{1\}}\sigma[\cdots],
\end{split}
\end{equation}
where $\sigma[\cdots]$
is obtained by applying the nontrivial permutation $\sigma$ of $(x,y,z)$
to the terms listed in $[\cdots]$.

Moreover,
we know that:
\begin{equation*}
\sum_{j,k,l\geq 1}\frac{\pd^3 F^{o,ext}}{\pd T_j\pd T_k\pd T_l}\bigg|_{\bm T =0} \cdot
x^{-j-1} y^{-k-1} z^{-l-1}
= G_{(3)}^o(x,y,z) - G_{(3)}^{\text{WK}}(x,y,z),
\end{equation*}
where $G_{(3)}^{\text{WK}}$ can be computed using \eqref{eq-thm-npt} and \eqref{eq-generating-WK-2}.
Then we plug the result into the above formula,
and finally obtain:
\begin{equation}
\begin{split}
&\sum_{j,k,l\geq 1}\frac{\pd^3 F^{o,ext}}{\pd T_j\pd T_k\pd T_l}\bigg|_{\bm T =0} \cdot
x^{-j-1} y^{-k-1} z^{-l-1}\\
=&\bigg[
\frac{a(x)a(-x)a(y)b(-y)b(-z)c(z)}{z(x^2-y^2)(y^2-z^2)}
+\frac{a(x)a(-x)a(-y)b(y)b(-z)c(z)}{z(x^2-y^2)(x^2-z^2)}\\
&+ \frac{a(x)a(-x)c(y)b(-z)a(-y)c(z)}{yz(x^2-z^2)}
+ \frac{a(x)a(-x)c(y)b(-z)a(-y)b(z)}{y(x^2-z^2)(y^2-z^2)} \\
&+\frac{a(-x)a(-y)a(-z)c(x)c(y)c(z)}{3xyz} +\frac{a(-x)a(-y)a(-z)b(y)b(z)c(x)}{2x(x^2-y^2)(x^2-z^2)} \\
&+\frac{(y^2+z^2-2x^2) a(-x)a(-y)a(-z)b(x)c(y)c(z)}{2yz(y^2-x^2)(z^2-x^2)} \bigg]
+\sum_{\sigma\in S_3\backslash\{1\}}\sigma[\cdots].
\end{split}
\end{equation}

\section{Proof of the Main Theorem}
\label{sec-proof}

In this section
we prove Theorem \ref{thm-mainthm} and the recursions given in \S \ref{sec-recursion}
using the Virasoro constraints.
Following the strategy developed in \cite{zhou3, zhou4},
we derive a reformulation of the Virasoro constraints in the fermionic picture.

\subsection{Computations of $a_{1,m}^o$}
\label{sec-a1}

Let $U^o\in Gr_{(0)}$ be the point corresponding to the tau-function $\tau^o$,
and let $\{a_{n,m}^o\}$ be its affine coordinates.
Then by definition,
$U^o$ is spanned by the following normalized basis:
\ben
\big\{z^{n+\half} + \sum_{m\geq 0} a_{n,m}^o z^{-m-\half}\big\}_{n\geq 0}.
\een

In \cite{al3},
Alexandrov has found an admissible basis for the point $U^o\in Gr_{(0)}$,
expressed in terms of some integrals.
One can try to expand these integrals and derive a normalized basis from this admissible basis.
Using this method,
it is easy to find the relations between $\{a_{n,m}^o\}_{n\geq 1}$ and $\{a_{n,m}^{\text{WK}}\}$,
but not easy to write down the numbers explicitly.
Here we follow another strategy.
We will start from the Virasoro constraints,
and use fermionic reformulation of the Virasoro constraints to
derive some recursion relations and constraints for the affine coordinates $\{a_{n,m}^o\}_{n,m\geq 0}$,
which enables us to solve these affine coordinates.
These recursions lead also to an explicit formula for $\{a_{0,m}^o\}$.
Besides the recursions from Virasoro constraints,
we will also need some initial data for these recursions,
and we will see that a good choice of initial data is $\{a_{1,m}^o\}_{m\geq 0}$.
A way to compute $\{a_{1,m}^o\}_{m\geq 0}$
is to compute the dual wave function and then use Lemma \ref{lem-tau-to-Gr}.
Here we directly use the following result of Alexandrov (see \cite[(3.16)]{al3}):
\ben
\Phi_2^o (z)=
\frac{z^{\frac{3}{2}}}{\sqrt{2\pi}}e^{-\frac{z^3}{3}}
\int dy\cdot \exp\bigg(-\frac{y^3}{3!}+\frac{yz^2}{2}\bigg),
\een
with a suitably chosen contour.
If one treats $\Phi_2^o$ as a formal Laurent series in $z$,
then $z^{\half}\cdot\Phi_2^o$ is the second vector
in an admissible basis
(i.e., the basis element whose leading term is $z^{\frac{3}{2}}$)
for $U^o$.
It is well-known that the asymptotic expansion of the above integral
is given by the Faber-Zagier series.
In fact, one has:
\begin{equation*}
\begin{split}
\int dy\cdot \exp\bigg(-\frac{y^3}{3!}+\frac{yz^2}{2}\bigg)=&
\int d y\cdot
\exp\bigg(\frac{1}{3}z^3 - \half  y^2 z -\frac{1}{6} y^3
\bigg)\\
=& e^{\frac{z^3}{3}}
\int d  y\cdot
\exp\bigg( - \half  y^2 z -\frac{1}{6} y^3
\bigg)\\
=& z^{-\half} e^{\frac{z^3}{3}}
\int d  y\cdot
e^{ - \half  y^2}
\sum_{n\geq 0} \frac{(  -\frac{1}{6} y^3
z^{-\frac{3}{2}}
)^n}{n!}.
\end{split}
\end{equation*}
Evaluating the Gaussian type integrals,
one obtains:
\ben
\Phi_2^o \sim
\sum_{m=0}^\infty \frac{(6m-1)!!}{36^m\cdot (2m)!}z^{-3m+1}.
\een
In particular,
there is only one term $z^{\frac{3}{2}}$ of non-negative degree in
the expansion of $z^\half \cdot\Phi_2(z)$,
thus $z^\half \cdot\Phi_2(z)$ is also the second vector of the normalized basis for $U^o$,
and thus:
\ben
z^{\frac{3}{2}}+
\sum_{m\geq 0} a_{1,m}^o z^{-m-\half}=z^\half \cdot
\sum_{m=0}^\infty \frac{(6m-1)!!}{36^m\cdot (2m)!}z^{-3m+1}.
\een
Then we have:
\be
\label{eq-coeff-a-1}
a_{1,3q-2}^o =
\frac{1}{36^q} \cdot \frac{(6q+1)!!}{(2q)!}\cdot \frac{1}{6q+1},
\qquad q\geq 1,
\ee
and $a_{1,m}=0$ if $m\not\equiv 1 (\text{mod }3)$.
Comparing this with \eqref{eq-WK-coeff-012}, we see:
\be
a_{1,3m-2}^o
=a_{0,3m-1}^{\text{WK}}.
\ee
This proves a special case of \eqref{eq-mainthm-relation}.

\subsection{A reformulation of the Virasoro operators}

Before reformulating the Virasoro constraints into the fermionic picture,
we need to modify the Virasoro operators \eqref{eq-Virasoro} for $n\geq 1$.

Recall that after the change of variables
\begin{equation*}
t_n=(2n+1)!!\cdot T_{2n+1},\qquad
s_n= 2^{n+1}\cdot (n+1)!\cdot T_{2n+2},
\end{equation*}
the Virasoro operators $\cL_n^{ext}$ for $n\geq 1$ are (after taking $u=1$):
\begin{equation*}
\begin{split}
\cL_n^{ext}= \frac{1}{2^{n+1}}\bigg(& -\frac{\pd}{\pd T_{2n+3}} +
\sum_{i\geq 0} (2i+1)T_{2i+1} \frac{\pd}{\pd T_{2n+2i+1}}
+\frac{1}{2}\sum_{i=0}^{n-1} \frac{\pd^2}{\pd T_{2i+1}\pd T_{2n-2i-1}}\\
&+ \sum_{i\geq 0} (2i+2)T_{2i+2} \frac{\pd}{\pd T_{2n+2i+2}}
+\frac{3}{2}(n+1)\frac{\pd}{\pd T_{2n}}\bigg).
\end{split}
\end{equation*}
Using the following result of Buryak \cite[\S 5.2]{bu}:
\ben
\frac{\pd^{k}}{\pd T_2^k}\tau^o = \frac{\pd}{\pd T_{2k}} \tau^o,
\qquad \forall k\geq 1,
\een
one has:
\begin{equation*}
\frac{3}{2}(n+1)\frac{\pd}{\pd T_{2n}}=
\half\sum_{i=1}^{n-1} \frac{\pd^2}{\pd T_{2i} \pd T_{2n-2i}}
+(n+2) \frac{\pd}{\pd T_{2k}}.
\end{equation*}
Then we can rewrite the constraints $\cL_n^{ext}(\tau^o)=0$ (for $n\geq 1$) as:
\be
\widetilde\cL_n^{ext}(\tau^o)=0,
\ee
where $\widetilde\cL_n^{ext}$ are defined by
(see also \cite[(3.24)]{al3}):
\be
\label{eq-Virasoro-modify}
\begin{split}
\widetilde\cL_n^{ext}=& -\frac{\pd}{\pd T_{2n+3}} +
\sum_{i\geq 0} (2i+1)T_{2i+1} \frac{\pd}{\pd T_{2n+2i+1}}
+\frac{1}{2}\sum_{i=0}^{n-1} \frac{\pd^2}{\pd T_{2i+1}\pd T_{2n-2i-1}}\\
&+ \sum_{i\geq 0} (2i+2)T_{2i+2} \frac{\pd}{\pd T_{2n+2i+2}}
+\half\sum_{i=1}^{n-1} \frac{\pd^2}{\pd T_{2i} \pd T_{2n-2i}}
+(n+2)\frac{\pd}{\pd T_{2n}}\\
=& -\frac{\pd}{\pd T_{2n+3}} +(n+2)\frac{\pd}{\pd T_{2n}}
+\half \sum_{k=1}^{2n-1} \frac{\pd^2}{\pd T_k \pd T_{2n-k}}
+\sum_{k=1}^\infty kT_k\frac{\pd}{\pd T_{2n+k}}.
\end{split}
\ee
In the following subsections we will use this version of Virasoro operators.

\subsection{Virasoro operators in the fermionic picture}
\label{sec-Virasoro-ferm1}

Now we convert the Virasoro constraints for $\tau^o$ into the fermionic picture.

First we need to do some preparations.
Define:
\be
\label{eq-gamma-fermi}
\alpha_n = \sum_{r+s=n} :\psi_r\psi_s^*:,
\qquad n\in \bZ,
\ee
where $r,s$ are half-integers and $\psi_r,\psi_s^*$ are the fermions.
By the boson-fermion correspondence,
they act on the bosonic Fock space by:
\be
\alpha_n = \begin{cases}
(-n)T_{-n}, & n<0;\\
0, & n=0;\\
\frac{\pd}{\pd T_n}, & n>0.
\end{cases}
\ee
Moreover, we denote:
\ben
\alpha(z) = :\psi(z)\psi^*(z): =\sum_{n\in \bZ}
\alpha_n z^{-n-1},
\een
where $\psi(z)$, $\psi^*(z)$ are the generating series of fermions:
\ben
\psi(z)=\sum_{r\in \bZ+\half} \psi_r z^{-r-\half},
\qquad
\psi^*(z)=\sum_{r\in \bZ+\half} \psi_r^* z^{-r-\half}.
\een
Then by Wick's theorem and some standard computations of operator product expansions,
one has:
\begin{equation*}
\alpha(z)\alpha(w)=
\frac{1}{(z-w)^2} + :\pd_w \psi(w)\psi^*(w):+:\pd_w \psi^* (w)\psi(w):+\cdots,
\end{equation*}
and it follows that:
\be
\sum_{a+b=n} :\alpha_a\alpha_b:=
\sum_{r+s=n} (s-r) :\psi_r\psi_s^*:,
\ee
where $a,b$ are integers and $r,s$ are half-integers.

Now we rewrite the Virasoro operators in terms of the fermions.
First by \eqref{eq-Virasoro} we can rewrite $\cL_{-1}^{ext}$ as:
\begin{equation*}
\begin{split}
\cL_{-1}^{ext} =& - \alpha_1 + \alpha_{-2}
+\half \sum_{a+b= -2} :\alpha_a \alpha_b:\\
=& -\sum_{r_1+s_1= 1} :\psi_{r_1} \psi_{s_1}^* :
+ \sum_{r_2+s_2= -2}(1+\frac{s_2-r_2}{2}) :\psi_{r_2} \psi_{s_2}^* :\\
=&-\bigg(\psi_\half \psi_\half^*
+\sum_{k=0}^\infty (\psi_{-k-\half}\psi_{k+\frac{3}{2}}^*-\psi_{-k-\half}^*\psi_{k+\frac{3}{2}})
\bigg)
+\half\psi_{-\half} \psi_{-\frac{3}{2}}^*
+\frac{3}{2}\psi_{-\frac{3}{2}}\psi_{-\half}^*\\
&+\sum_{l=0}^\infty \bigg(
(l+\frac{5}{2})\psi_{-l-\frac{5}{2}}\psi_{l+\frac{1}{2}}^*
+(l+\frac{1}{2})\psi_{-l-\frac{5}{2}}^*\psi_{l+\frac{1}{2}})
\bigg).
\end{split}
\end{equation*}
Similarly, we have:
\begin{equation*}
\begin{split}
\cL_0^{ext}=& - \alpha_3
+\half\sum_{a+b=0} :\alpha_a \alpha_b:
+\frac{13}{8} \\
=& - \bigg( \psi_{\half}\psi_{\frac{5}{2}}^*+
\psi_{\frac{3}{2}}\psi_{\frac{3}{2}}^* + \psi_{\frac{5}{2}}\psi_{\frac{1}{2}}^*
+\sum_{k=0}^\infty (\psi_{-k-\half}\psi_{k+\frac{7}{2}}^* - \psi_{-k-\half}^*\psi_{k+\frac{7}{2}})
\bigg)\\
&+ \sum_{l=0}^\infty (l+\half)(
\psi_{-l-\half}\psi_{l+\half}^* +\psi_{-l-\half}^*\psi_{l+\half})
+\frac{13}{8},
\end{split}
\end{equation*}
and for $n\geq 1$,
\begin{equation*}
\begin{split}
\widetilde\cL_n^{ext}=&
- \alpha_{2n+3} +(n+2) \alpha_{2n}
+\half\sum_{a+b=2n} :\alpha_a \alpha_b: \\
=& -\bigg( \sum_{k=0}^{2n+2} \psi_{k+\half} \psi_{2n-k+\frac{5}{2}}^*
+\sum_{k=0}^\infty (\psi_{-k-\half}\psi_{2n+k+\frac{7}{2}}^* -\psi_{-k-\half}^*\psi_{2n+k+\frac{7}{2}})
\bigg) \\
&+\sum_{k=0}^\infty \bigg(
(2n+k+\frac{5}{2}) \psi_{-k-\half}\psi_{2n+k+\half}^*
-(-k+\frac{3}{2})\psi_{-k-\half}^* \psi_{2n+k+\half}
\bigg)\\
&+\sum_{k=0}^{2n-1}(2n-k+\frac{3}{2}) \psi_{k+\half} \psi_{2n-k-\half}^*.
\end{split}
\end{equation*}

\subsection{Recursions for affine coordinates from Virasoro constraints}
\label{sec-Virasoro-ferm2}

Now we derive the recursion relations for the affine coordinates $\{a_{n,m}^o\}$
using the fermionic reformulations of Virasoro operators.

First we compute $e^{-A^o}\cL_{-1}^{ext} e^{A^o} |0\rangle$ using
the properties \eqref{eq-conjugate-1} and \eqref{eq-conjugate-2}.
For example,
the first term is:
\begin{equation*}
\begin{split}
e^{-A^o}\psi_\half \psi_\half^* e^{A^o} |0\rangle=&
\big(e^{-A^o}\psi_\half e^{A^o}\big)\big(e^{-A^o} \psi_\half^* e^{A^o} \big)|0\rangle\\
=&\bigg(\psi_\half - \sum_{m\geq 0} a_{0,m}^o\psi_{-m-\half}\bigg)
\bigg(\psi_\half^* + \sum_{n\geq 0} a_{n,0}^o\psi_{-n-\half}^* \bigg)
|0\rangle \\
=& \bigg( a_{0,0}^o-\sum_{m,n\geq 0}a_{0,m}^o a_{n,0}^o\psi_{-m-\half}\psi_{-n-\half}^*
\bigg) |0\rangle.
\end{split}
\end{equation*}
Similarly,
we can compute $e^{-A^o}\cL_{-1}^{ext}e^{A^o}|0\rangle$ term by term,
and the result is:
\be
\label{eq-stringeqn-fermion}
\begin{split}
&e^{-A^o}\cL_{-1}^{ext} e^{A^o}|0\rangle\\
=& \bigg(
-a_{0,0}^o +\sum_{m,n\geq 0} a_{0,m}^oa_{n,0}^o \psi_{-m-\half}\psi_{-n-\half}^*
+\half \psi_{-\half}\psi_{-\frac{3}{2}}^*
 + \frac{3}{2}\psi_{-\frac{3}{2}}\psi_{-\half}^*\\
&-\sum_{k,n\geq 0} a_{n,k+1}^o \psi_{-k-\half}\psi_{-n-\half}^*
+\sum_{k,m\geq 0} a_{k+1,m}^o \psi_{-m-\half} \psi_{-k-\half}^* \\
&+\sum_{l,n\geq 0}
(l+\frac{5}{2})a_{n,l}^o \psi_{-l-\frac{5}{2}}\psi_{-n-\half}^*
+\sum_{l,m\geq 0}(l+\frac{1}{2})a_{l,m}^o\psi_{-m-\half} \psi_{-l-\frac{5}{2}}^*
\bigg)|0\rangle.
\end{split}
\ee
Since the boson-fermion correspondence is an isomorphism,
by $\cL_{-1}^{ext}(\tau^o)=0$ we must have:
\ben
\cL_{-1}^{ext} e^{A^o} |0\rangle=\cL_{-1}^{ext} |U^o\rangle=0,
\een
and thus the expression \eqref{eq-stringeqn-fermion} is also zero.
Now computing the coefficient of each term in \eqref{eq-stringeqn-fermion},
we may obtain a family of constraints for $\{a_{n,m}^o\}$.
For example,
the coefficient of $|0\rangle$ gives:
\be
a_{0,0}^o=0.
\ee
Then we compute the coefficient of $\psi_{-\half}\psi_{-\half^*}$ and get:
\ben
-a_{0,1}^o+a_{1,0}^o=0.
\een
We already know that $a_{1,0}^o=0$,
thus:
\be
a_{0,1}^o=0.
\ee
Next,
computing the coefficient of $\psi_{-\half}\psi_{-\frac{3}{2}}^*|0\rangle$ gives:
\ben
\half-a_{1,1}^o+a_{2,0}^o=0,
\een
and computing the coefficient of $\psi_{-\frac{3}{2}}\psi_{-\frac{1}{2}}^*|0\rangle$ gives:
\ben
\frac{3}{2}-a_{0,2}^o+a_{1,1}^o=0.
\een
By \eqref{eq-coeff-a-1} we have already known that $a_{1,1}^o=\frac{5}{24}$,
thus:
\be
a_{2,0}^o=-\frac{7}{24},
\qquad\qquad
a_{0,2}^o=\frac{41}{24}.
\ee
Similarly,
we compute the coefficient of $\psi_{-m-\half}\psi_{-n-\frac{1}{2}}^*|0\rangle$ for general $n,m\geq 0$,
and in this way we easily obtain:
\be
\label{eq-rec-amn-0}
a_{0,m}^o a_{n,0}^o
-a_{n,m+1}^o+ a_{n+1,m}^o
+(m+\half)a_{n,m-2}^o
+(n-\frac{3}{2})a_{n-2,m}^o =0,
\ee
for $m+n\geq 2$.
Here we denote $a_{n,m}^o :=0$ if $n<0$ or $m<0$.
This proves the $1$-step recursion \eqref{eq-1step-rec}.
In the following subsections we will see that
the affine coordinates $\{a_{n,m}^o \}_{n,m\geq 0}$ can be uniquely determined using
the above $1$-step recursion and the initial values \eqref{eq-coeff-a-1}.
Nevertheless,
here we write down all the constraints provided by $\widetilde\cL_n^{ext}$ for $n\geq 0$ for completeness.
We omit the details and only give the results since the computations are similar to the case of $\cL_{-1}^{ext}$.
The constraint from the operator $\cL_0^{ext}$ is:
\begin{equation*}
\begin{split}
0=& \bigg(
-a_{0,2}^o-a_{1,1}^o-a_{2,0}^o +\sum_{m,n\geq 0} (a_{0,m}^oa_{n,2}^o+a_{1,m}^oa_{n,1}^o+a_{2,m}^oa_{n,0}^o)
\psi_{-m-\half}\psi_{-n-\half}^*\\
&-\sum_{k,n\geq 0} a_{n,k+3}^o \psi_{-k-\half}\psi_{-n-\half}^*
-\sum_{k,m\geq 0} a_{k+3,m}^o \psi_{-k-\half}^* \psi_{-m-\half} \\
&+\sum_{l,n\geq 0}
(l+\frac{1}{2})a_{n,l}^o \psi_{-l-\frac{1}{2}}\psi_{-n-\half}^*
-\sum_{l,m\geq 0}(l+\frac{1}{2})a_{l,m}^o \psi_{-l-\frac{1}{2}}^* \psi_{-m-\half}
+\frac{13}{8}
\bigg)|0\rangle,
\end{split}
\end{equation*}
and the constraints from $\widetilde\cL_n^{ext}$ (see \eqref{eq-Virasoro-modify}) are:
\begin{equation*}
\begin{split}
0=&\bigg(-\sum_{k=0}^{2n+2} a_{k,2n+2-k}^o  +
\sum_{k=0}^{2n+2}\sum_{m,l\geq 0} a_{k,m}^o  a_{l,2n+2-k}^o \psi_{-m-\half}\psi_{-l-\half}^*\\
&-\sum_{k,l\geq 0} a_{l,2n+3+k}^o  \psi_{-k-\half}\psi_{-l-\half}^*
-\sum_{k,m\geq 0} a_{2n+3+k,m}^o  \psi_{-k-\half}^*\psi_{-m-\half}\\
&+\sum_{k=0}^{2n-1} (2n+\frac{3}{2}-k)(a_{k,2n-1-k}^o -\sum_{m,l\geq 0} a_{k,m}^o a_{l,2n-1-k}^o
\psi_{-m-\half}\psi_{-l-\half}^*)\\
&+\sum_{k,l\geq 0}(2n+\frac{5}{2}+k) a_{l,2n+k}^o \psi_{-k-\half}\psi_{-l-\half}^*\\
&-\sum_{k,m\geq 0}(k-\frac{3}{2})a_{2n+k,m}^o \psi_{-k-\half}^*\psi_{-m-\half}\bigg)
|0\rangle,
\end{split}
\end{equation*}
for every $n\geq 1$.
Now computing the coefficients of $|0\rangle$ and $\psi_{-m-\half}\psi_{-n-\half}^*|0\rangle$
for every $n,m\geq 0$,
we have proved the recursion relations and constraints in
Theorem \ref{thm-coeff-rec-1} and Theorem \ref{thm-coeff-rec-2}.

\subsection{The modulo $3$ property}

Recall that we already have
(see \S \ref{sec-a1} and \S \ref{sec-Virasoro-ferm2}):
\ben
a_{0,0}^o = a_{0,1}^o = a_{1,0}^o =0.
\een
Let $k$ be a fixed positive integer,
and assume that we already know the values of all $a_{n,m}^o$ with $n+m<k$.
Then using the recursion \eqref{eq-rec-amn-0} and
the initial value $a_{1,k-1}^o$ given by \eqref{eq-coeff-a-1},
we can compute all $a_{n,m}^o$ with $n+m=k$.
Thus by induction
all the coordinates $\{a_{n,m}^o\}_{n,m\geq 0}$ are uniquely determined by
the $1$-step recursion \eqref{eq-rec-amn-0} and initial data \eqref{eq-coeff-a-1}.

Moreover,
similar to the case of Witten-Kontsevich tau-function,
here we also have the following modulo $3$ property:
\begin{Lemma}
We have $a_{n,m}^o=0$ for $n+m\not\equiv -1(\text{mod }3)$.
\end{Lemma}
\begin{proof}
We prove by induction on $n+m$.
The cases $m+n=0,1$ are clear.

Now assume that the conclusion holds for all $(m,n)$ with $m+n\leq k$.
If $k=3l$,
then consider the case $m+n=3l+1$.
In this case \eqref{eq-rec-amn-0} gives:
\ben
&&a_{0,3l-n}^o a_{n,0}^o
-a_{n,3l+1-n}^o+ a_{n+1,3l-n}^o\\
&&+(3l-n+\half)a_{n,3l-n-2}^o
+(n-\frac{3}{2})a_{n-2,3l-n}^o =0
\een
By induction hypothesis we easily see $a_{n,3l-n-2}^o=a_{n-2,3l-n}^o=0$.
Moreover,
we must have either $n\not\equiv -1(\text{mod }3)$ or $(3l-n)\not\equiv -1(\text{mod }3)$,
thus $a_{0,3l-n}^o a_{n,0}^o=0$.
Then the above recursion gives:
\ben
a_{n,3l+1-n}^o= a_{n+1,3l-n}^o.
\een
Recall that we already have $a_{1,3l}^o=0$,
thus $a_{n,3l+1-n}^o = 0$ for every $n$.

Similarly for $k=3l-1$ or $k=3l-2$,
we can also verify the case for $m+n\leq k+1$ by induction.
\end{proof}

Now using the above property,
we can rewrite the recursion \eqref{eq-rec-amn-0} as follows:
\be
\label{eq-rec-amn}
\begin{split}
&a_{p,3q-1-p}^o- a_{p+1,3q-2-p}^o\\
=&a_{0,3q-2-p}^o a_{p,0}^o+(3q-\frac{3}{2}-p)a_{p,3q-4-p}^o
+(p-\frac{3}{2})a_{p-2,3q-2-p}^o,
\end{split}
\ee
for every $q\geq 2$ and $0\leq p\leq 3q-1$.

\subsection{Computations of $a_{2,m}^o$}

In this subsection we compute $a_{2,m}^o$ for general $m\geq 0$.
First recall that
\be
a_{2,0}^o=-\frac{7}{24},
\ee
and $a_{2,m}^o=0$ for $m\not\equiv 0 (\text{mod }3)$.
By taking $p=1$ in the recursion \eqref{eq-rec-amn},
we obtain:
\be
\label{eq-rec-1and2}
a_{2,3q-3}^o = a_{1,3q-2}^o
-(3q-\frac{5}{2})a_{1,3q-5}^o.
\ee
for $q\geq 2$.
For example,
\ben
&&a_{2,3}^o= a_{1,4}^o-\frac{7}{2}a_{1,1}^o
=\frac{385}{1152}-\frac{7}{2}\cdot\frac{5}{24}=
-\frac{455}{1152},\\
&&a_{2,6}^o= a_{1,7}^o-\frac{13}{2}a_{1,4}^o
=\frac{85085}{82944}-\frac{13}{2}\cdot\frac{385}{1152}=
-\frac{95095}{82944},\\
&&a_{2,9}^o= a_{1,10}^o-\frac{19}{2}a_{1,7}^o
=\frac{ 37182145}{7962624}-\frac{19}{2}\cdot\frac{ 85085}{82944}=
-\frac{40415375}{7962624},\\
&&\cdots\cdots
\een

Now plugging \eqref{eq-coeff-a-1} into the relation \eqref{eq-rec-1and2},
one easily finds that:
\be
\label{eq-coeff-a-2}
a_{2,3q-3}^o
=-\frac{1}{36^q}\cdot \frac{(6q+1)!!}{(2q)!} \cdot \frac{1}{6q-1}
=a_{1,3q-2}^{\text{WK}},
\ee
given by the Faber-Zagier series.
This proves another special case of \eqref{eq-mainthm-relation}.

\subsection{Computations of $a_{0,m}^o$}

Now we compute $a_{0,3q-1}^o$.
First recall that we have already known:
\be
a_{0,2}^o = \frac{41}{24}.
\ee
And then taking $p=0$ in \eqref{eq-rec-amn},
we obtain a recursion:
\be
\label{eq-coeff-rec-0}
a_{0,3q-1}^o=
a_{1,3q-2}^o+(3q-\frac{3}{2})a_{0,3q-4}^o,
\qquad q\geq 2,
\ee
where $\{a_{1,3q-2}^o\}$ are given by \eqref{eq-coeff-a-1}.
For example,
\ben
&&a_{0,5}^o=a_{1,4}^o+\frac{9}{2}a_{0,2}^o=
\frac{385}{1152}+\frac{9}{2}\cdot\frac{41}{24}=
\frac{9241}{1152},\\
&&a_{0,8}^o=a_{1,7}^o+\frac{15}{2}a_{0,5}^o=
\frac{85085}{82944}+\frac{15}{2}\cdot\frac{9241}{1152}=
\frac{5075225}{82944},\\
&&a_{0,11}^o=a_{1,10}^o+\frac{21}{2}a_{0,8}^o=
\frac{ 37182145}{7962624}+\frac{21}{2}\cdot\frac{5075225}{82944}=
\frac{5153008945}{7962624},\\
&&\cdots\cdots
\een

Now we claim that for every $q\geq 1$, we have:
\be
\label{eq-a0m}
a_{0,3q-1}^o=\big(\frac{3}{2}\big)^q \cdot (2q-1)!!\cdot \sum_{j=0}^q
\frac{(6j-1)!!}{54^j\cdot (2j)!\cdot (2j-1)!!},
\ee
where we use the convention $(-1)!!:=1$.
We prove by induction on $q$.
The case $q=1$ can be easily checked.
Now assume the conclusion holds for $q$,
then by \eqref{eq-coeff-a-1} and the recursion \eqref{eq-coeff-rec-0}
one has:
\begin{equation*}
\begin{split}
a_{0,3q+2}^o=&
a_{1,3q+1}^o+(3q+\frac{3}{2})a_{0,3q-1}^o\\
=&\frac{(6q+5)!!}{36^{q+1}\cdot (2q+2)!}
+\frac{6q+3}{2}\big(\frac{3}{2}\big)^q (2q-1)!!
\sum_{j=0}^q \frac{(6j-1)!!}{54^j\cdot (2j)!\cdot (2j-1)!!}\\
=&\big(\frac{3}{2}\big)^{q+1} \cdot(2q+1)!!\cdot
\sum_{j=0}^{q+1} \frac{(6j-1)!!}{54^j\cdot (2j)!\cdot (2j-1)!!},
\end{split}
\end{equation*}
Thus the conclusion holds,
and this proves \eqref{eq-a0m-explicit}.

\subsection{Computations of $a_{n,m}^o$ for $n\geq 3$}

Now we can complete the proof of Theorem \ref{thm-mainthm}.
In this proof
we will make use of the following recursion relation for the affine coordinates $\{a_{n,m}^{\text{WK}}\}_{n,m\geq 0}$
of Witten-Kontsevich tau-function:
\begin{Lemma}
[\cite{zhou4}]
For every $m+n\geq 2$,
\be
\label{eq-rec-WK}
a_{n+1,m}^{\text{WK}}=
a_{n,m+1}^{\text{WK}} - a_{0,m}^{\text{WK}}a_{n,0}^{\text{WK}}
-\bigg(
(m-\half)a_{n,m-2}^{\text{WK}} + (n-\half)a_{n-2,m}^{\text{WK}}
\bigg).
\ee
\end{Lemma}
This recursion is obtained by converting the string equation for the Witten-Kontsevich tau-function
into the fermionic picture,
see \cite[(61)]{zhou4}.
Notice that the notations $A_{m,n}$ in that work are related to $a_{n,m}^{\text{WK}}$ by:
\ben
A_{m,n}=\bigg(-\frac{\sqrt{-2}}{4}\bigg)^{\frac{m+n+1}{3}}
\cdot a_{n,m}^{\text{WK}},
\qquad
m,n\geq 0.
\een

Now we can prove \eqref{eq-mainthm-relation} by induction on $p$.
First notice that by comparing \eqref{eq-coeff-a-1}, \eqref{eq-coeff-a-2}
with \eqref{eq-WK-coeff-012},
we already have:
\be
\label{eq-relation-12}
a_{1,m}^o=a_{0,m+1}^{\text{WK}},
\qquad
a_{2,m}^o=a_{1,m+1}^{\text{WK}}.
\ee

First consider the case $a_{n,m}^o$ with $n=3p+1$ and $m=3q-2$.
By the recursion \eqref{eq-rec-amn} we have:
\be
a_{3p,3q-1}^o-a_{3p+1,3q-2}^o
=(3q-\frac{3}{2})a_{3p,3q-4}^o
+(3p-\frac{3}{2})a_{3p-2,3q-2}^o .
\ee
Thus by induction hypothesis:
\begin{equation*}
\begin{split}
a_{3p+1,3q-2}^o
=&a_{3p,3q-1}^o
-(3q-\frac{3}{2})a_{3p,3q-4}^o-(3p-\frac{3}{2})a_{3p-2,3q-2}^o\\
=& a_{3p-1,3q}^{\text{WK}} -a_{0,3q-1}^o a_{3p-1,0}^{\text{WK}}
-(3q-\frac{3}{2})(a_{3p-1,3q-3}^{\text{WK}}-a_{0,3q-4}^o a_{3p-1,0}^{\text{WK}})\\
&- (3p-\frac{3}{2})a_{3p-3,3q-1}^{\text{WK}}\\
=&a_{3p-1,3q}^{\text{WK}} -a_{0,3q-1}^o a_{3p-1,0}^{\text{WK}}
 +(3q-\frac{3}{2}) a_{0,3q-4}^o a_{3p-1,0}^{\text{WK}}\\
&-\bigg(
(3q-\frac{3}{2}) a_{3p-1,3q-3}^{\text{WK}} + (3p-\frac{3}{2}) a_{3p-3,3q-1}^{\text{WK}}
\bigg).
\end{split}
\end{equation*}
By \eqref{eq-rec-WK} we know that:
\begin{equation*}
(3q-\frac{3}{2}) a_{3p-1,3q-3}^{\text{WK}} + (3p-\frac{3}{2}) a_{3p-3,3q-1}^{\text{WK}}
=a_{3p-1,3q}^{\text{WK}} -a_{0,3q-1}^{\text{WK}}a_{3p-1,0}^{\text{WK}} -a_{3p,3q-1}^{\text{WK}},
\end{equation*}
and thus
\begin{equation*}
\begin{split}
a_{3p+1,3q-2}^o
=&a_{3p-1,3q}^{\text{WK}} -a_{0,3q-1}^o a_{3p-1,0}^{\text{WK}}
+(3q-\frac{3}{2}) a_{0,3q-4}^o a_{3p-1,0}^{\text{WK}}\\
&-(a_{3p-1,3q}^{\text{WK}} -a_{0,3q-1}^{\text{WK}}a_{3p-1,0}^{\text{WK}} -a_{3p,3q-1}^{\text{WK}})\\
=&a_{3p,3q-1}^{\text{WK}} +a_{3p-1,0}^{\text{WK}} \bigg(
-a_{0,3q-1}^{\text{WK}}+a_{0,3q-1}^o-(3q-\frac{3}{2})a_{0,3q-4}^o\bigg).
\end{split}
\end{equation*}
Moreover, by \eqref{eq-coeff-rec-0} and \eqref{eq-relation-12} we have:
\begin{equation*}
-a_{0,3q-1}^{\text{WK}}+ a_{0,3q-1}^o-(3q-\frac{3}{2})a_{0,3q-4}^o=
-a_{1,3q-2}^o+a_{1,3q-2}^o=0,
\end{equation*}
therefore:
\be
a_{3p+1,3q-2}^o=a_{3p,3q-1}^{\text{WK}}.
\ee

Similarly,
for $n=3p+2$ one can check that:
\begin{equation*}
\begin{split}
a_{3p+2,3q-3}^o=&
a_{3p+1,3q-2}^o-(3q-\frac{5}{2})a_{3p+1,3q-5}^o-(3p-\half)a_{3p-1,3q-3}^o\\
=&a_{3p,3q-1}^{\text{WK}} -\bigg((3q-\frac{5}{2})a_{3p,3q-4}^{\text{WK}}+(3p-\half)a_{3p-2,3q-2}^{\text{WK}}\bigg)\\
=&a_{3p,3q-1}^{\text{WK}}+(a_{3p+1,3q-2}^{\text{WK}}-a_{3p,3q-1}^{\text{WK}})\\
=&a_{3p+1,3q-2}^{\text{WK}},
\end{split}
\end{equation*}
by \eqref{eq-rec-amn}, \eqref{eq-rec-WK} and the induction hypothesis.

Finally consider the case $n=3p$.
We have:
\begin{equation*}
\begin{split}
a_{3p,3q-1}^o=& -a_{0,3q-1}^oa_{3p-1,0}^o+
a_{3p-1,3q}^o-(3q-\frac{1}{2})a_{3p-1,3q-3}^o-(3p-\frac{5}{2})a_{3p-3,3q-1}^o\\
=& -a_{0,3q-1}^oa_{3p-1,0}^o+ a_{3p-2,3q+1}^{\text{WK}}
-(3q-\frac{1}{2})a_{3p-2,3q-2}^{\text{WK}}\\
&-(3p-\frac{5}{2})(a_{3p-4,3q}^{\text{WK}}-a_{0,3q-1}^oa_{3p-4,0}^{\text{WK}})\\
=&-a_{0,3q-1}^oa_{3p-1,0}^o+
a_{3p-2,3q+1}^{\text{WK}}+ (3p-\frac{5}{2})a_{0,3q-1}^oa_{3p-4,0}^{\text{WK}}\\
&+(a_{3p-1,3q}^{\text{WK}}-a_{3p-2,3q+1}^{\text{WK}})\\
=& a_{3p-1,3q}^{\text{WK}} - a_{0,3q-1}^o\bigg(a_{3p-1,0}^o- (3p-\frac{5}{2})a_{3p-4,0}^{\text{WK}}\bigg),
\end{split}
\end{equation*}
thus it suffices to prove that:
\begin{equation*}
a_{3p-1,0}^o-(3p-\frac{5}{2})a_{3p-4,0}^{\text{WK}}
=a_{3p-1,0}^{\text{WK}}.
\end{equation*}
By the induction hypothesis we only need to prove:
\begin{equation*}
a_{3p-2,1}^{\text{WK}}-(3p-\frac{5}{2})a_{3p-4,0}^{\text{WK}}
=a_{3p-1,0}^{\text{WK}},
\end{equation*}
and this follows from \eqref{eq-rec-WK}.
Now we have finished the proof of Theorem \ref{thm-mainthm}.

\vspace{.2in}

{\em Acknowledgements}.
The author thanks Prof. Jian Zhou for guidance in this subject.
And he thanks Prof. Huijun Fan for encouragement,
Dr. Qingsheng Zhang for helpful discussions,
and Dr. Giulio Ruzza for helpful comments.
He also thanks the anonymous referee for suggestions.

\end{document}